\documentclass[twocolumn]{autart}
\usepackage{eurosym}
\usepackage[round]{natbib}
\usepackage{graphicx}
\usepackage{graphpap}
\usepackage{amsmath}
\usepackage{amsfonts}
\usepackage{graphicx}
\usepackage{amssymb}
\usepackage{color}
\usepackage{xcolor}
\usepackage{algorithm}
\usepackage{algpseudocode}

\allowdisplaybreaks[3]
\definecolor{light}{gray}{.65}
\input{tcilatex}
\begin{document}

\begin{frontmatter}

\title{System representations in subspaces of finite-sample signals and their application to data-driven fault detection}

\thanks[footnoteinfo]{This work has been supported by the National Natural Science Foundation of
China under Grant 62322303 and 62233012.}

\author[linlin]{Linlin Li}, \ead{linlin.li@ustb.edu.cn}
%\author[Cui]{Kaixin Cui}, %\ead{kaixincui@bit.edu.cn}
%\author[Cui]{Dawei Shi},  %\ead{daweishi@bit.edu.cn}
\author[Ding]{Steven X. Ding},  \ead{steven.ding@uni-due.de}
\author[linlin]{Jiahao Wang}, \ead{jiahaowang@xs.ustb.edu.cn}
\author[Zhong]{Maiying Zhong}, \ead{myzhong@sdust.edu.cn}
\author[Ding]{Wei Cheng}  \ead{cheng.wei@uni-due.de}

\address[linlin]{School of Automation and Electrical Engineering, University of Science and Technology Beijing, 100083 Beijing, China}
\address[Ding]{Institute for Automatic Control and Complex Systems, University of Duisburg-Essen, 47057 Duisburg, Germany}
\address[Zhong]{College of Electrical Engineering and Automation,
			Shandong University of Science and Technology, Qingdao 266590, China}

                             % keyword list or with the
                                          % help of the Automatica
                                          % keyword wizard

\begin{abstract}                          % Abstract of not more than 200 words.
This paper deals with system representations in finite-sample signal subspaces and their application to data-driven fault detection. The first part addresses concepts of finite-sample image and kernel system representations and, associated with them, image and residual subspaces of finite-sample signals. On this basis, the equivalence between the fundamental lemma and finite-sample image subspace is demonstrated. While the image representation models the nominal system dynamics, the residual representation describes uncertainties in the input-output data and is essential for fault detection. This result extends the fundamental lemma and builds the basis for exploring data-driven fault detection. In the second part, a data-driven projection-based fault detection approach is developed. By means of a singular value decomposition, orthogonal projections onto the image and residual subspaces are realized in the context of a low-rank matrix approximation, leading to projection-based residual generation and evaluation. Finally, analysis of detection performance in the framework of matrix perturbation theory and comparison with existing data-driven fault detection methods are explored.
\end{abstract}

\end{frontmatter}

\section{Introduction}

Fault detection is a fundamental problem in the monitoring and supervision
of dynamic control systems, with critical importance in applications.
Undetected faults may lead to significant performance degradation,
constraint violations, or even loss of closed-loop stability. Consequently,
the development of reliable methods for the timely detection of faults
remains a central topic in automatic control systems %
\citep{annurev-control-survey}. Conventional fault detection approaches are
predominantly model-based, relying on accurate mathematical descriptions of
the underlying system input-output causality and dynamics %
\citep{PFC00,Ding2008}. Nevertheless, the dominat methodology in recent
years is data-driven, particularly machine learning (ML) based methods %
\citep{JinWang2021,ZHAO2026Survey}. Undoubtedly, in the era of information,
digitalization and big-data, this trend is a logic consequence. Due to the
distinguishing technological difference between these two classes of
methods, first-principles models vs. process operation data, the
methodological difference receives less attention. Existing data-driven
fault detection methods for dynamic control systems mainly rely on a
data-driven reconstruction of the system output %
\citep{Ding_IJP_2014,Hao2023,CHEN2023Automatica}. Methodologically
differently, ML-based data-driven methods characterize the system behavior
in terms of a lower-dimensional feature space or latent-variable
representations learned from historical observations. To capture the nominal
system behavior implicitly from the measurement data, which span a
higher-dimensional subspace due to system uncertainties, projection methods
serve as a capable tool \citep{JinWang2021}. Inspired by this methodological
practice, a model-based fault detection scheme has been proposed in our
recent work using an orthogonal projection \citep{DL2026}. In this
framework, the input-output data space is decomposed into an image subspace
representing the nominal system behavior and a residual subspace spanned by
uncertainties. Projecting online input-output data onto the image subspace
enables us to detect faults in the residual subspace. This approach is
particularly capable of dealing with multiplicative uncertainties and faults
that span the whole input-output data space. A direct transfer of this
approach to a data-driven implementation is hindered due to the lack of a
firm theoretic foundation. Specifically, in the model-based projection
framework, coprime factorization theory serves as the mathematical tool to
define system image, kernel and residual representations and subspaces in $%
\ell _{2}\left[ 0,\infty \right) $ Hilbert space, and construct projection
operators by means of image and kernel representations. Although in our
previous work, data-driven image and kernel representations were addressed %
\citep{Ding_automatica2014}, their realization remains in the setting of
input-output dynamics.

Willems' behavioral systems theory describes a system as the set of all
admissible trajectories and focuses primarily on the nominal dynamics %
\citep{Willems1998}. The fundamental lemma \citep{WILLEMS2005} provides a
non-parametric characterization of all finite-length trajectories of an
unknown linear time-invariant system through Hankel matrices constructed
from persistently exciting data, thereby offering a model-free description
of the system image (nominal) behavior. Although the fundamental lemma as a
cornerstone is widely applied for the development of advanced data-driven
control methods \citep{Waarde2020TAC,markovsky2021,DorflerTAC2023}, it can
be observed that limited research leverages it to address fault detection
issues. During preparing this work, we notice probably the first publication
on the fundamental lemma-based fault detection \citep{MARKOVSKY2026}, in
which an orthogonal projection-based residual generation and evaluation are
realized using the nominal data set representing the system image behavior.
It is remarkable that while advanced fault detection methods are strongly
focused on fault detection in uncertain systems \citep{Ding2020}, the
fundamental lemma-based control and detection methods either rely on nominal
data set without explicitly imposing uncertainties induced subspace, or
address noises as uncertainties that are then estimated in an input-output
dynamic setting \citep{WANGautomatica2025}.

Motivated by the aforementioned observations and analysis, the endeavors in
this paper are devoted to two objectives. Firstly, analogous to the image
and kernel representations in $\ell _{2}\left[ 0,\infty \right) $ Hilbert
space, finite-sample image and kernel representations are investigated and,
on this basis, image and residual subspaces of finite-sample signals are
introduced. Furthermore, relations of these concepts to the fundamental
lemma are explored. In the second part, a data-driven projection-based fault
detection approach is then developed on the assumption of availability of
input-output data that are corrupted with uncertainties and span the space
of finite-sample input-output data. The main contributions of this work are,
(i) introduction of image, kernel representations and, corresponding to
them, image and residual subspaces of finite-sample signals, a theoretic
basis for subsequent research, (ii) presentation of an alternative proof and
extension of the fundamental lemma by taking into account uncertainties in
data, (iii) development of a data-driven fault detection approach by means
of orthogonal projections, and (iv) analysis of detection performance of the
proposed detection system and comparison with the existing data-driven
detection methods.

The rest of the paper is organized as follows. The necessary preliminaries
and problem formulation are introduced in Section 2. Section 3 explores the
theoretical foundation of this work, including image and kernel
representations, the corresponding subspaces of finite-sample signals, and
discusses their relations to the fundamental lemma. It is followed by the
development of a data-driven projection-based fault detection approach in
Section 4.

\textbf{Notation.} Throughout this paper, standard notation known in control
theory and linear algebra is adopted. In addition, $\mathbb{Z}$ $\left( 
\mathbb{Z}_{\geq 0}\right) $ denotes the set of all integers (greater or
equal to zero). Given $i,k,q,s,N\in \mathbb{Z}_{\geq 0},\phi (i)\in \mathbb{R%
}^{q},i\in \left[ k+1,k+N\right] ,$ 
\begin{gather*}
\phi _{s}(k)=\left[ 
\begin{array}{c}
\phi (k+1) \\ 
\vdots  \\ 
\phi (k+s)%
\end{array}%
\right] \in \mathbb{R}^{sq}, \\
\mathcal{H}_{s}(\phi _{\lbrack k+1:k+N]})\hspace{-2pt}=\hspace{-2pt}\left[ 
\hspace{-2pt}%
\begin{array}{ccc}
\phi _{s}(k) & \cdots  & \hspace{-2pt}\phi _{s}(k\hspace{-2pt}+\hspace{-2pt}N%
\hspace{-2pt}-\hspace{-2pt}s)%
\end{array}%
\hspace{-2pt}\right] \hspace{-4pt}\in \hspace{-2pt}\mathbb{R}^{sq\times
\left( N-s+1\right) }
\end{gather*}%
denote the stacked vector and Hankel matrix of $\phi (k),$ respectively. For
a sequence of blocks $G_{k},k\in \mathbb{Z},\mathcal{T}_{q,t}\left( G\right) 
$ denotes a Toeplitz matrix defined by 
\begin{equation*}
\left( \mathcal{T}_{q,t}\left( G\right) \right) _{i,j}=G_{l+i-j},i=1,\cdots
,q,j=1,\cdots ,t,l\in \mathbb{Z}_{\geq 0}.
\end{equation*}%
A typical example from our subsequent work is%
\begin{gather}
\left( \mathcal{T}_{s,s}\left( G\right) \right) _{i,j}=G_{i-j}=\left\{ 
\begin{array}{l}
CA^{i-j}B,i>j \\ 
D,i=j \\ 
0,i<j,%
\end{array}%
\right.   \label{eq1-1} \\
\mathcal{T}_{s,s}\left( G\right) =\left[ 
\begin{array}{cccc}
D & 0 & \cdots  & 0 \\ 
CB & \text{ }D & \ddots  & \vdots  \\ 
\vdots  & \ddots  & \ddots  & 0 \\ 
CA^{s-2}B & \cdots  & \text{ }CB & \text{ }D%
\end{array}%
\right] .  \label{eq1-2}
\end{gather}%
$\mathcal{O}_{s}$ and $\mathcal{C}_{s}$ denote matrices 
\begin{equation}
\setlength{\abovedisplayskip}{6pt}\setlength{\belowdisplayskip}{0pt}\mathcal{%
O}_{s}=\left[ 
\begin{array}{c}
C \\ 
\vdots  \\ 
CA^{s-1}%
\end{array}%
\right] ,\mathcal{C}_{s}=\left[ 
\begin{array}{ccc}
B & \text{ }\cdots  & \text{ }A^{s-1}B%
\end{array}%
\right] .  \label{eq1-3}
\end{equation}

\textbf{Remark 1} \textit{We adopt notation typically used in the fault
detection framework, which may slightly differ from the standard one adopted
in behavior systems theory and related study of data-driven system analysis
and control}.

\section{Preliminaries and problem formulation}

\label{sec:preliminaries}

In this section, system coprime factorizations and the associated
alternative system representations, basic fault detection methods and the
fundamental lemma are briefly introduced, followed by a problem formulation
addressed in this work.

\subsection{System factorizations, image, kernel representations, and signal
subspaces}

Consider a linear time invariant (LTI) system $G(z)$, whose minimal state
space realization is described by%
\begin{align}
x(k+1)& =Ax(k)+Bu(k),  \label{eq2-10} \\
y(k)& =Cx(k)+Du(k),  \label{eq2-11}
\end{align}%
where $u\in \mathbb{R}^{p},x\in \mathbb{R}^{n},y\in \mathbb{R}^{m}$ are the
system input, state, and output vectors, respectively. The left and right
coprime factorizations (LCF and RCF) of $G(z),$ 
\begin{equation*}
\setlength{\abovedisplayskip}{6pt}\setlength{\belowdisplayskip}{6pt} G(z)=%
\hat{M}^{-1}(z)\hat{N}(z)=N(z)M^{-1}(z),
\end{equation*}%
are alternative system representation forms. Associated with them, there
exist right and left coprime pairs (RCP and LCP), $\left( \hat{X}(z),\hat{Y}%
(z)\right) $ and $\left( X(z),Y(z)\right) $ so that the double Bezout
identity holds \citep{Vinnicombe-book},%
\begin{equation}
\setlength{\abovedisplayskip}{6pt}\setlength{\belowdisplayskip}{6pt} \left[ 
\begin{array}{cc}
X(z) & \text{ }Y(z) \\ 
-\hat{N}(z) & \text{ }\hat{M}(z)%
\end{array}%
\right] \left[ 
\begin{array}{cc}
M(z) & \text{ }-\hat{Y}(z) \\ 
N(z) & \text{ }\hat{X}(z)%
\end{array}%
\right] =I.  \label{eq2-6}
\end{equation}%
The state space models of the LCPs and RCPs $\left( \hat{M},\hat{N}\right)
,\left( X,Y\right) $ and $\left( M,N\right) ,\left( \hat{X},\hat{Y}\right) $
are as follows 
\begin{align}
\left[ 
\begin{array}{cc}
\hat{M} & \text{ }\hat{N}%
\end{array}%
\right] & =\left( A_{L},\left[ 
\begin{array}{cc}
-L & B_{L}%
\end{array}%
\right] ,C,\left[ 
\begin{array}{cc}
I & D%
\end{array}%
\right] \right) ,  \label{eq2-1}
\end{align}%
\begin{align}
\left[ 
\begin{array}{cc}
X & \text{ }Y%
\end{array}%
\right] & =\left( A_{L},\left[ 
\begin{array}{cc}
-B_{L} & -L%
\end{array}%
\right] ,F,\left[ 
\begin{array}{cc}
I & 0%
\end{array}%
\right] \right) ,  \notag \\
\left[ 
\begin{array}{c}
M \\ 
N%
\end{array}%
\right] & =\left( A_{F},B,\left[ 
\begin{array}{c}
F \\ 
C_{F}%
\end{array}%
\right] ,\left[ 
\begin{array}{c}
I \\ 
D%
\end{array}%
\right] \right) ,  \label{eq2-1a} \\
\left[ 
\begin{array}{c}
\hat{Y} \\ 
\hat{X}%
\end{array}%
\right] & =\left( A_{F},-L,\left[ 
\begin{array}{c}
F \\ 
-C_{F}%
\end{array}%
\right] ,\left[ 
\begin{array}{c}
0 \\ 
I%
\end{array}%
\right] \right) ,  \label{eq2-2}
\end{align}%
where $B_{L}=B-LD,C_{F}=C+DF$, and matrices $F$ and $L$ are designed such
that $A_{F}=A+BF$ and $A_{L}=A-LC$ are Schur matrices.

\textbf{Remark 2} \textit{Hereafter, we may drop the domain variable }$z$%
\textit{\ or }$k$\textit{\ when there is no risk of confusion}.

On account of the state space models of $\left( \hat{M},\hat{N}\right)
,\left( X,Y\right) ,\left( M,N\right) $ and $\left( \hat{X},\hat{Y}\right) $
as well as the Bezout identity (\ref{eq2-6}), the system input-output data $%
\left( u,y\right) $ are expressed by 
\begin{equation}
\setlength{\abovedisplayskip}{6pt}\setlength{\belowdisplayskip}{6pt} \left[ 
\begin{array}{c}
u \\ 
y%
\end{array}%
\right] =\left[ 
\begin{array}{c}
M \\ 
N%
\end{array}%
\right] v+\left[ 
\begin{array}{c}
-\hat{Y} \\ 
\hat{X}%
\end{array}%
\right] r,  \label{eq2-4}
\end{equation}%
which can be interpreted as a control loop with an observer-based state
feedback controller, $u=F\hat{x}+v.$ The state estimate $\hat{x}$ is
delivered by a state observer that serves as a residual generator as well %
\citep{Ding2020}, 
\begin{gather}
\hat{x}(k+1)=A_{L}\hat{x}(k)+B_{L}u(k)+Ly(k),  \label{eq2-9} \\
r(k)=y(k)-\hat{y}(k),\hat{y}(k)=C\hat{x}(k)+Du(k).  \label{eq2-9c}
\end{gather}%
Signal $r$ is called residual and comprises exclusively information about
any uncertainties in the system, like unknown inputs, model mismatching and
even faults. It is noteworthy that the Bezout identity (\ref{eq2-6})
introduces a one-to-one mapping between $\left( u,y\right) $ and $\left(
v,r\right) ,$%
\begin{equation}
\setlength{\abovedisplayskip}{6pt}\setlength{\belowdisplayskip}{6pt} \left[ 
\begin{array}{c}
v \\ 
r%
\end{array}%
\right] =\left[ 
\begin{array}{cc}
X & \text{ }Y \\ 
-\hat{N} & \text{ }\hat{M}%
\end{array}%
\right] \left[ 
\begin{array}{c}
u \\ 
y%
\end{array}%
\right] ,  \label{eq2-3}
\end{equation}%
which allows us to define the system image and residual subspaces with $%
\left( v,r\right) $ as the latent (hidden) variables. The residual generator
(\ref{eq2-9})-(\ref{eq2-9c}) is also called stable kernel representation of $%
G$ and plays a core role in observer-based fault diagnosis \citep{Ding2020}.
In control theory, systems $\left( I_{G},I_{C}\right) ,$ 
\begin{equation}
\setlength{\abovedisplayskip}{6pt}\setlength{\belowdisplayskip}{6pt} I_{G}=%
\left[ 
\begin{array}{c}
M \\ 
N%
\end{array}%
\right] ,I_{C}=\left[ 
\begin{array}{c}
-\hat{Y} \\ 
\hat{X}%
\end{array}%
\right] ,
\end{equation}%
are called image representations of $G$ and feedback controller $C$,
respectively, where $u=Cy=-\hat{Y}\hat{X}^{-1}y$ \citep{Vinnicombe-book}.
Corresponding to them, the $\mathcal{H}_{2}$ subspaces $\mathcal{I}_{G}$ and 
$\mathcal{R}_{G}$, 
\begin{gather}
\mathcal{I}_{G}=\left\{ \left[ 
\begin{array}{c}
u \\ 
y%
\end{array}%
\right] :\left[ 
\begin{array}{c}
u \\ 
y%
\end{array}%
\right] \hspace{-2pt}=\hspace{-2pt}\left[ 
\begin{array}{c}
M \\ 
N%
\end{array}%
\right] v,v\in \mathcal{H}_{2}^{p}\right\} ,  \label{eq3-12a} \\
\mathcal{R}_{G}=\left\{ \left[ 
\begin{array}{c}
u \\ 
y%
\end{array}%
\right] :\left[ 
\begin{array}{c}
u \\ 
y%
\end{array}%
\right] =\left[ 
\begin{array}{c}
-\hat{Y} \\ 
\hat{X}%
\end{array}%
\right] r,r\in \mathcal{H}_{2}^{m}\right\} ,  \label{eq3-13a}
\end{gather}%
are defined and called image and residual subspaces, respectively.
Attributed to the Bezout identity (\ref{eq2-6}), $\mathcal{R}_{G}$ is the
complementary subspace of $\mathcal{I}_{G}$ and $\mathcal{H}_{2}^{p+m}=%
\mathcal{I}_{G}\oplus \mathcal{R}_{G}.$ Signals in $\mathcal{I}_{G}$
represent the nominal system dynamics, while signals in $\mathcal{R}_{G}$
describe the system uncertain dynamics that cause $r\neq 0.$

\subsection{Fault detection methods \label{Subsection2-2}}

It is known that all observer-based residual generators can be parameterized
by 
\begin{equation}
\setlength{\abovedisplayskip}{6pt}\setlength{\belowdisplayskip}{6pt}%
r=R\left( \hat{M}y-\hat{N}u\right)  \label{eq2-12}
\end{equation}%
with a stable post-filter $R$ as the parameterisation system \citep{Ding2008}%
. A further classical fault detection technique is the parity space method %
\citep{WillskyAUTO76}. A parity space-based residual generator is a finite
impulse response (FIR) system of the form 
\begin{gather}
P_{s}\left( y_{s}(k)-\mathcal{T}_{s,s}\left( G\right) u_{s}(k)\right)
\label{eq2-14} \\
=\left[ 
\begin{array}{cc}
-P_{s}\mathcal{T}_{s,s}\left( G\right) & \text{ }P_{s}%
\end{array}%
\right] \left[ 
\begin{array}{c}
u_{s}(k) \\ 
y_{s}(k)%
\end{array}%
\right] \in \mathbb{R}^{l},  \label{eq2-14a}
\end{gather}%
where $\mathcal{T}_{s,s}\left( G\right) $ is given in (\ref{eq1-1}), $1\leq
l\leq sm-n,$ and $P_{s}\left( \neq 0\right) \in \mathbb{R}^{l\times sm}$ is
called parity matrix that solves the equation%
\begin{equation}
\setlength{\abovedisplayskip}{6pt}\setlength{\belowdisplayskip}{6pt}P_{s}%
\mathcal{O}_{s}=0.  \label{eq2-15}
\end{equation}%
The integer $s$ is called the order of the parity matrix. It is worth
mentioning that a data-driven variation of parity space method was developed
and, based on it, the data-driven kernel representation-based fault
detection was proposed in \citep{Ding_automatica2014}.

Recently, an orthogonal projection-based fault detection has been proposed %
\citep{DL2026}. The core of this method is an orthogonal projection of $%
(u,y) $ onto the image subspace ${\mathcal{I}}_{G}$ by an operator $\mathcal{%
P}_{{\mathcal{I}}_{G}}$ and the generation of a residual vector $r_{{%
\mathcal{I}}_{G}}$ by means of 
\begin{equation*}
\setlength{\abovedisplayskip}{6pt}\setlength{\belowdisplayskip}{6pt} r_{{%
\mathcal{I}}_{G}}=\left[ 
\begin{array}{c}
u \\ 
y%
\end{array}%
\right] -\mathcal{P}_{{\mathcal{I}}_{G}}\left[ 
\begin{array}{c}
u \\ 
y%
\end{array}%
\right] =\left( \mathcal{I}-\mathcal{P}_{{\mathcal{I}}_{G}}\right) \left[ 
\begin{array}{c}
u \\ 
y%
\end{array}%
\right] .
\end{equation*}%
It is known from operator theory \citep{Kato_book} that 
\begin{equation*}
\func{dist}\left( \left[ 
\begin{array}{c}
u \\ 
y%
\end{array}%
\right] ,{\mathcal{I}}_{G}\right) =\left\Vert \left( \mathcal{I}-\mathcal{P}%
_{{\mathcal{I}}_{G}}\right) \left[ 
\begin{array}{c}
u \\ 
y%
\end{array}%
\right] \right\Vert _{2},
\end{equation*}%
namely $\left\Vert r_{{\mathcal{I}}_{G}}\right\Vert _{2}$ measures the
(minimal) distance from $(u,y)$ to ${\mathcal{I}}_{G}.$ Hence, $\left\Vert
r_{{\mathcal{I}}_{G}}\right\Vert _{2}$ serves as a residual evaluation
function. On this basis, gap metric \citep{Kato_book} that measures the
similarity between two subspaces is introduced for the threshold setting.

\subsection{Fundamental lemma}

In this subsection, we briefly review the fundamental lemma and some
relevant results \citep{WILLEMS2005,markovsky2021}.

\begin{definition}
\label{Definition 1} The system input $u_{s}(k)\in \mathbb{R}^{sp}$ is
persistently exciting of order $s$ if the associated Hankel matrix $\mathcal{%
H}_{s}(u_{[k+1:k+N]})$ is of full row rank, i.e.%
\begin{equation*}
rank\left( \mathcal{H}_{s}(u_{[k+1:k+N]})\right) =sp.
\end{equation*}%
\end{definition}

\begin{lemma}
\label{Lemma2-1} Let $\left( u_{s}(k),y_{s}(k)\right) $ denote $s$-samples
long input-output trajectories of the LTI system (\ref{eq2-10})-(\ref{eq2-11}%
) over the horizon $\left[ k_{0}+1,k_{0}+N\right] $ and assume $u$ to be
persistently exciting of order $s+n$. Then, any trajectory $\left(
u_{s}(k),y_{s}(k)\right) $ of the system can be expressed by  
\begin{equation}
\begin{bmatrix}
\,u_{s}(k)\, \\ 
\,y_{s}(k)\,%
\end{bmatrix}%
=T_{G,s}g,T_{G,s}=%
\begin{bmatrix}
\,\mathcal{H}_{s}(u_{[k_{0}+1:k_{0}+N]})\, \\ 
\,\mathcal{H}_{s}(y_{[k_{0}+1:k_{0}+N]})\,%
\end{bmatrix}%
,  \label{eq2-8}
\end{equation}%
for some $g\in \mathbb{R}^{N-s+1}$, and 
\begin{equation*}
\setlength{\abovedisplayskip}{6pt}\setlength{\belowdisplayskip}{6pt}
\forall g\in \mathbb{R}^{N-s+1},%
\begin{bmatrix}
\,u_{s}(k)\, \\ 
\,y_{s}(k)\,%
\end{bmatrix}%
=T_{G,s}g
\end{equation*}%
is a trajectory of the system.
\end{lemma}

Lemma \ref{Lemma2-1} is known as fundamental lemma, which can be
equivalently expressed by $\func{Im}\left( T_{G,s}\right) =\mathcal{B}%
_{G,s}, $%
\begin{equation}
\setlength{\abovedisplayskip}{6pt}\setlength{\belowdisplayskip}{6pt}\mathcal{%
B}_{G,s}=\left\{ \left[ 
\begin{array}{c}
u_{s}(k) \\ 
y_{s}(k)%
\end{array}%
\right] :\left( u,y\right) \text{ is a trajectory of }G\right\} .
\end{equation}%
In the data-driven framework, the concept of persistent excitation and the
fundamental lemma are the cornerstone.

\subsection{Problem formulation}

In our work, it is assumed that data are collected during fault-free
operations under persistent excitation, whose dynamics are governed by 
\begin{align}
x(k+1)& =Ax(k)+Bu(k)+\Delta _{p}(k),  \label{eq2-20} \\
y(k)& =Cx(k)+Du(k)+\Delta _{s}(k),  \label{eq2-21}
\end{align}%
where $\Delta _{p},\Delta _{s}$ represent uncertainties in the process and
sensors, respectively. They are assumed to be unknown but $\ell _{2}$%
-bounded. We also consider a special case 
\begin{gather}
\left[ 
\begin{array}{c}
\Delta _{p}(k) \\ 
\Delta _{s}(k)%
\end{array}%
\right] =\left[ 
\begin{array}{c}
\omega (k) \\ 
v(k)%
\end{array}%
\right] \sim \mathcal{N}\left( 0,\Sigma \right) ,  \label{eq2-20a} \\
\Sigma =\mathbb{E}\left( \left[ 
\begin{array}{c}
w(i) \\ 
v(i)%
\end{array}%
\right] \left[ 
\begin{array}{c}
w(j) \\ 
v(j)%
\end{array}%
\right] ^{T}\right) =\left[ 
\begin{array}{cc}
\Sigma _{\omega } & S_{\omega v} \\ 
S_{\omega v}^{T} & \Sigma _{v}%
\end{array}%
\right] \delta _{ij}  \label{eq2-21b}
\end{gather}%
with the process and measurement noises $\left( w,v\right) $ that are
uncorrelated with the system input. Here, $\delta _{ij}=1$ for $i=j,\delta
_{ij}=0$ for $i\neq j.$ Note that the system input-output dynamics can be
equivalently written as%
\begin{equation}
\setlength{\abovedisplayskip}{6pt}\setlength{\belowdisplayskip}{6pt}\left\{ 
\begin{array}{l}
\hat{x}(k+1)=A\hat{x}(k)+Bu(k)+Lr(k) \\ 
y(k)=\hat{y}(k)+r(k),%
\end{array}%
\right.  \label{eq2-25}
\end{equation}%
which is called kernel-based model. Hereby, the dynamics of the residual $r$
is described by%
\begin{equation}
\setlength{\abovedisplayskip}{6pt}\setlength{\belowdisplayskip}{6pt}%
r=C\left( zI-A_{L}\right) ^{-1}\left( \Delta _{p}-L\Delta _{s}\right)
+\Delta _{s}.
\end{equation}%
It is assumed that $span\left\{ r(k),k\in \mathbb{Z}_{\geq 0}\right\} =%
\mathbb{R}^{m}.$

In the sequel, analogous to $I_{G}$ and $I_{C}$ and the input-output
representation (\ref{eq2-4}), we first address the image representations of
finite-sample signals. Specifically, the input-output data $\left(
u_{s}(k),y_{s}(k)\right) $ are described by an algebraic system driven by $%
\left( v,r\right) $ over a (finite) time interval. Whereby, the response of $%
\left( u_{s}(k),y_{s}(k)\right) $ to $v$, called image representation,
describes the nominal dynamics, while the response to $r$ characterizes the
dynamics induced by uncertainties. On this basis, image and residual
subspaces of finite-sample signals, similar to $\left( \mathcal{I}_{G},%
\mathcal{R}_{G}\right) ,$ are defined. An interesting application of this
work is to demonstrate the equivalence of the image subspace to $\func{Im}%
(T_{G,s}),$ implying an alternative proof of the fundamental lemma. The
response to $r$ plays a core role in fault detection. Differing from (\ref%
{eq2-3}) in the model-based framework, the kernel representation in this
work is introduced as the null-space of the image representation. In the
data-driven setting, this study extends the fundamental lemma to deal with
data corrupted with uncertainties. These theoretical results enable the
development of a data-driven projection-based fault detection approach. On
the basis of a low-rank matrix approximation by means of a singular value
decomposition (SVD) of the data matrix, projection-based residual generation
and evaluation are realized. Finally, detection performance of the proposed
approach is analyzed in the framework of matrix perturbation theory and
compared with existing detection methods.

\section{Image and kernel representations, signal subspaces, and fundamental
lemma}

\label{sec:method}

In this section, we introduce finite-sample image and kernel representations
and, associated with them, the image and residual subspaces as alternative
system model forms. We make a concerted effort to explore relations between
(i) the image subspace and fundamental lemma and (ii) residual and parity
subspaces. They highlight some unexplored theoretic aspects and lay the
groundwork for our endeavours to develop data-driven fault detection methods.

\subsection{Image representation and subspace, fundamental lemma}

In \citep{Ding_automatica2014}, data-driven image and kernel representations
have been introduced and applied for the purpose of fault detection and
control. In this subsection, we firstly introduce a finite-sample image
representation, an analogous form of the image representation\ (\ref{eq2-1a}%
) in finite samples-long signal subspace, that generalizes the data-driven
representation in \citep{Ding_automatica2014}. On this basis, image
subspace, a data-driven image representation and its equivalence to the
fundamental lemma are explored. We begin with the following lemma given in %
\citep{Ding_automatica2014}.

\begin{lemma}
\label{Lemma3-1}Given the system (\ref{eq2-10})-(\ref{eq2-11}), there
exist a matrix pair $\left( M_{s},N_{s}\right) $ and a vector $%
v_{s+n}(k-n)\in \mathbb{R}^{\left( s+n\right) p},$ s.t.%
\begin{gather}
\left[ 
\begin{array}{c}
u_{s}(k) \\ 
y_{s}(k)%
\end{array}%
\right] =\left[ 
\begin{array}{c}
M_{s} \\ 
N_{s}%
\end{array}%
\right] v_{s+n}(k-n),  \label{eq3-1} \\
\left[ 
\begin{array}{c}
M_{s} \\ 
N_{s}%
\end{array}%
\right] =\left[ 
\begin{array}{cc}
H_{u,x} & \text{ }H_{u,v} \\ 
H_{y,x} & \text{ }H_{y,v}%
\end{array}%
\right] \in \mathbb{R}^{s(p+m)\times (s+n)p},  \label{eq3-1a}
\end{gather}%
where matrices $H_{u,x},H_{u,v},H_{y,x}$ and $H_{y,v}$ are%
\begin{gather}
H_{u,x}=\left[ 
\begin{array}{c}
F \\ 
\vdots  \\ 
FA_{F}^{s-1}%
\end{array}%
\right] \left[ 
\begin{array}{ccc}
A_{F}^{n-1}B & \cdots  & \text{ }B%
\end{array}%
\right] \in \mathbb{R}^{sp\times np},  \notag \\
H_{u,v}=\left[ 
\begin{array}{cccc}
I & 0 & \cdots  & \text{ }0 \\ 
FB & I & \ddots  & \vdots  \\ 
\vdots  & \ddots  & \ddots  & \text{ }0 \\ 
FA_{F}^{s-2}B & \cdots  & FB & \text{ }I%
\end{array}%
\right] \in \mathbb{R}^{sp\times sp},  \notag \\
H_{y,x}=\left[ 
\begin{array}{c}
C_{F} \\ 
\vdots  \\ 
C_{F}A_{F}^{s-1}%
\end{array}%
\right] \left[ 
\begin{array}{ccc}
A_{F}^{n-1}B & \cdots  & \text{ }B%
\end{array}%
\right] \in \mathbb{R}^{sm\times np},  \notag \\
H_{y,v}=\left[ 
\begin{array}{cccc}
D & 0 & \cdots  & \text{ }0 \\ 
C_{F}B & \text{ }D & \ddots  & \vdots  \\ 
\vdots  & \ddots  & \ddots  & \text{ }0 \\ 
C_{F}A_{F}^{s-2}B & \cdots  & C_{F}B & \text{ }D%
\end{array}%
\right] \in \mathbb{R}^{sm\times sp},  \notag
\end{gather}%
with $F=F_{d},$ and all eigenvalues of matrix $A_{F}=A+BF_{d}$ equal zero. 
\end{lemma}

The idea behind the proof of this lemma is to set $u=F_{d}x+v$ in the image
representation (\ref{eq2-1a}). The proof is straightforward on account of
the deadbeat behavior, i.e. $A_{F_{d}}^{n}=0$. The following lemma is
essential to prove that the model (\ref{eq3-1}) holds for any gain matrix $F$%
.

\begin{lemma}
\label{Lemma3-2}Given the system (\ref{eq2-10})-(\ref{eq2-11}), and suppose
that RCPs $(M_{1},N_{1})$ and $(M_{2},N_{2})$ are constructed for $F=F_{1}$
and $F=F_{2}\neq F_{1}$, respectively. Then, it holds%
\begin{align}
\begin{bmatrix}
\,u\, \\ 
\,y\,%
\end{bmatrix}%
& =%
\begin{bmatrix}
\,M_{1}\, \\ 
\,N_{1}\,%
\end{bmatrix}%
v=%
\begin{bmatrix}
\,M_{2}\, \\ 
\,N_{2}\,%
\end{bmatrix}%
Vv,  \label{eq3-2} \\
V& =\left( A_{F_{1}},B,F_{1}-F_{2},I\right) .
\end{align}%
Correspondingly,%
\begin{gather}
\begin{bmatrix}
M_{1,s}\, \\ 
N_{1,s}\,%
\end{bmatrix}%
=%
\begin{bmatrix}
\,M_{2,s}\, \\ 
\,N_{2,s}\,%
\end{bmatrix}%
V_{s+n},V_{s+n}=\mathcal{T}_{s+n,s+n}\left( V\right) ,  \label{eq3-3} \\
\left( \mathcal{T}_{s+n,s+n}\left( V\right) \right) _{i,j}=V_{i-j}=\left\{ 
\begin{array}{l}
(F_{1}-F_{2})A_{F_{1}}^{i-j}B,i>j \\ 
I,i=j \\ 
0,i<j,%
\end{array}%
\right.   \notag
\end{gather}%
where $\left( M_{i,s},N_{i,s}\right) ,i=1,2,$ are $\left( M_{s},N_{s}\right) 
$ given in Lemma \ref{Lemma3-1} for $F=F_{i}.$ 
\end{lemma}%
\begin{proof}
The relation (\ref{eq3-2}) is a known result, see,
for instance, \citep{Ding2020}. Below, (\ref{eq3-3}) is proven. For
notational simplicity, Toeplitz matrices of $\left( M_{s},N_{s}\right) ,$\
for $F=F_{i},C_{F}=C_{F_{i}},A_{F}=A_{F_{i}},i=1,2,$%
\begin{gather}
M_{s,i}=\mathcal{T}_{s,s+n}\left( M^{i}\right) ,M_{k}^{i}=\left\{ 
\begin{array}{l}
F_{i}A_{F_{i}}^{k-1}B,k>0 \\ 
I,k=0 \\ 
0,k<0,%
\end{array}%
\right.  \label{eq5-2} \\
N_{s,i}=\mathcal{T}_{s,s+n}\left( N^{i}\right) ,N_{k}^{i}=\left\{ 
\begin{array}{l}
C_{F_{i}}A_{F_{i}}^{k-1}B,k>0 \\ 
D,k=0 \\ 
0,k<0,%
\end{array}%
\right.  \label{eq5-3} \\
\left( \mathcal{T}_{s,s+n}\left( M^{i}\right) \right)
_{l,j}=M_{n+l-j}^{i},\left( \mathcal{T}_{s,s+n}\left( N^{i}\right) \right)
_{l,j}=N_{n+l-j}^{i},  \notag
\end{gather}%
are used. Observe that%
\begin{gather*}
M_{2,s}V_{s+n}=\mathcal{T}_{s,s+n}\left( M^{2}\right) \mathcal{T}%
_{s+n,s+n}\left( V\right) =:\mathcal{T}_{s,s+n}\left( \bar{M}\right) , \\
\left( \mathcal{T}_{s,s+n}\left( \bar{M}\right) \right) _{l,j}=\bar{M}%
_{n+l-j}, \\
\bar{M}_{k}=\left\{ 
\begin{array}{l}
\dsum\limits_{l=0}^{k-2}F_{2}A_{F_{2}}^{l}B\left( F_{1}-F_{2}\right)
A_{F_{1}}^{k-2-l}B \\ 
+\left( F_{1}-F_{2}\right) A_{F_{1}}^{k-1}B+F_{2}A_{F_{2}}^{k-1}B,k>0 \\ 
I,k=0 \\ 
0,k<0.%
\end{array}%
\right.
\end{gather*}%
Attributed to the equalities,%
\begin{gather}
B\left( F_{1}-F_{2}\right) =A_{F_{1}}-A_{F_{2}}\Longrightarrow  \notag \\
A_{F_{1}}^{k-1}-A_{F_{2}}^{k-1}=\dsum\limits_{i=0}^{k-2}A_{F_{2}}^{i}B\left(
F_{1}-F_{2}\right) A_{F_{1}}^{k-i-2},  \label{eq3-5}
\end{gather}%
it turns out, by routine computations, that 
\begin{equation*}
\bar{M}_{k}=M_{k}^{1}\Longrightarrow M_{2,s}V_{s+n}=M_{1,s}.
\end{equation*}%
Similarly, we have 
\begin{gather*}
\mathcal{T}_{s,s+n}\left( \bar{N}\right) :=N_{2,s}V_{s+n},\left( \mathcal{T}%
_{s,s+n}\left( \bar{N}\right) \right) _{l,j}=\bar{N}_{n+l-j}, \\
\bar{N}_{k}=\left\{ 
\begin{array}{l}
\dsum\limits_{l=0}^{k-2}C_{F_{2}}A_{F_{2}}^{l}B\left( F_{1}-F_{2}\right)
A_{F_{1}}^{k-2-l}B \\ 
+D\left( F_{1}-F_{2}\right) A_{F_{1}}^{k-1}B+C_{F_{2}}A_{F_{2}}^{k-1}B,k>0
\\ 
D,k=0 \\ 
0,k<0,%
\end{array}%
\right.
\end{gather*}%
which, by means of (\ref{eq3-5}), equals 
\begin{equation*}
\bar{N}_{k}=\left\{ 
\begin{array}{l}
C_{F_{1}}A_{F_{1}}^{k-1}B,k>0 \\ 
D,k=0 \\ 
0,k<0,%
\end{array}%
\right. =N_{k}^{1}.
\end{equation*}%
Thus, $N_{2,s}V_{s+n}=N_{1,s}.$ The proof is completed.
\end{proof}

On the basis of Lemmas \ref{Lemma3-1}-\ref{Lemma3-2}, we have the following
theorem.

\begin{theorem}
\label{Theorem 3-1}Given the system (\ref{eq2-10})-(\ref{eq2-11}), any
input-output data $\left( u_{s}(k),y_{s}(k)\right) $ can be modelled by (\ref%
{eq3-1}) for some $v_{s+n}(k-n)\in \mathbb{R}^{\left( s+n\right) p}$and $%
A_{F}=A+BF$ for an arbitrary gain matrix $F$. Furthermore, for $s\geq n,$%
\begin{equation}
\setlength{\abovedisplayskip}{3pt}\setlength{\belowdisplayskip}{0pt}
rank\left[ 
\begin{array}{c}
M_{s} \\ 
N_{s}%
\end{array}%
\right] =sp+n.  \label{eq3-3a}
\end{equation}
\end{theorem}
\begin{proof}
Consider the image representation of the system (\ref{eq2-10}%
)-(\ref{eq2-11}) given by (\ref{eq2-1a}), which, according to (\ref{eq3-2})
in Lemma \ref{Lemma3-2}, is further written into 
\begin{equation*}
\setlength{\abovedisplayskip}{4pt}\setlength{\belowdisplayskip}{4pt}\left[ 
\begin{array}{c}
M \\ 
N%
\end{array}%
\right] =\left[ 
\begin{array}{c}
M_{d} \\ 
N_{d}%
\end{array}%
\right] V\Longrightarrow \left[ 
\begin{array}{c}
u \\ 
y%
\end{array}%
\right] =\left[ 
\begin{array}{c}
M_{d} \\ 
N_{d}%
\end{array}%
\right] \bar{v},\bar{v}=Vv,
\end{equation*}%
where $\left( M_{d},N_{d}\right) $\ denote the image representation for $%
F=F_{d}.$ It follows from the state space model of $\bar{v}=Vv,$\ 
\begin{align*}
x_{V}(k+1)& =A_{F}x_{V}(k)+Bv(k),x_{V}(0)=0, \\
\bar{v}(k)& =\left( F-F_{d}\right) x_{V}(k)+v(k),
\end{align*}%
that $\bar{v}_{s+n}\left( k-n\right) =V_{s+n}v_{s+n}\left( k-n\right) .$\
Consequently, according to Lemmas \ref{Lemma3-1}-\ref{Lemma3-2}, we have 
\begin{equation*}
\setlength{\abovedisplayskip}{6pt}\setlength{\belowdisplayskip}{6pt}%
\begin{bmatrix}
\,u_{s}(k)\, \\ 
\,y_{s}(k)\,%
\end{bmatrix}%
=\left[ 
\begin{array}{c}
M_{d,s} \\ 
N_{d,s}%
\end{array}%
\right] \bar{v}_{s+n}\left( k-n\right) =\left[ 
\begin{array}{c}
M_{s} \\ 
N_{s}%
\end{array}%
\right] v_{s+n}\left( k-n\right) .
\end{equation*}%
Note that for $F=0,$ (\ref{eq3-1a}) is subject to 
\begin{equation*}
\setlength{\abovedisplayskip}{4pt}\setlength{\belowdisplayskip}{4pt}\left[ 
\begin{array}{cc}
H_{u,x} & \text{ }H_{u,v} \\ 
H_{y,x} & \text{ }H_{y,v}%
\end{array}%
\right] =\left[ 
\begin{array}{cc}
0 & \text{ }I \\ 
\mathcal{O}_{s}\mathcal{C}_{n} & \text{ }H_{y,v}%
\end{array}%
\right] \Longrightarrow \text{\textit{(\ref{eq3-3a}).}}
\end{equation*}%
Hence, according to Lemma \ref{Lemma3-2}, (\ref{eq3-3a}) holds for any $F.$
The proof is completed.
\end{proof}

An immediate result of Theorem \ref{Theorem 3-1} is that any input-output
data $\left( u_{s}(k),y_{s}(k)\right) $ belongs to the subspace $\mathcal{I}%
_{G,s},$ 
\begin{equation}
\setlength{\abovedisplayskip}{6pt}\setlength{\belowdisplayskip}{6pt} 
\mathcal{I}_{G,s}=\left\{ \left[ 
\begin{array}{c}
u_{s}(k) \\ 
y_{s}(k)%
\end{array}%
\right] ,\left[ 
\begin{array}{c}
u_{s}(k) \\ 
y_{s}(k)%
\end{array}%
\right] =\left[ 
\begin{array}{c}
M_{s} \\ 
N_{s}%
\end{array}%
\right] h\right\}  \label{eq3-4}
\end{equation}%
for $h\in \mathbb{R}^{\left( s+n\right) p}.$ $\mathcal{I}_{G,s}$ is called
image subspace.

\begin{corollary}
\label{Co3-1}Given the system (\ref{eq2-10})-(\ref{eq2-11}) with the input that
is persistently exciting of order $s+n$, then $\exists \mathcal{H}%
_{s+n}\left( v_{\left[ k_{p}+1:k_{p}+N\right] }\right) \in \mathbb{R}%
^{(s+n)p\times \left( N-s+1\right) },k_{p}=k_{0}-n$ such that%
\begin{gather}
T_{G,s}=\left[ 
\begin{array}{c}
M_{s} \\ 
N_{s}%
\end{array}%
\right] \mathcal{H}_{s+n}\left( v_{\left[ k_{p}+1:k_{p}+N\right] }\right) ,
\label{eq3-6} \\
\func{Im}\left( T_{G,s}\right) =\mathcal{B}_{G,s}=\mathcal{I}_{G,s}.
\label{eq3-6a}
\end{gather}
\end{corollary}

\begin{proof}
Equation (\ref{eq3-6}) is attributed to Theorem \ref{Theorem
3-1}, while (\ref{eq3-6a}) follows from Lemma \ref{Lemma2-1}, Theorem \ref%
{Theorem 3-1} and (\ref{eq3-6}).
\end{proof}

As highlighted in \citep{markovsky2021}, the original proof of the
fundamental lemma in \textit{\citep{WILLEMS2005} "is based on a kernel
representation of the system and uses the notion of annihilators of the
behavior}". Also the proof with a state space representation in %
\citep{Waarde2020} is based on the concept of annihilators of the behavior.
Particularly, quoting from \citep{markovsky2021}: "\textit{Currently there
is no constructive proof that gives an intuition why the additional
persistency of excitation is needed nor how conservative the conditions are}%
." The results in Lemmas \ref{Lemma3-1}-\ref{Lemma3-2}, Theorem \ref{Theorem
3-1} and Corollary \ref{Co3-1} suggest a constructive proof of the
fundamental lemma on the basis of system image representation and an
alternative interpretation in the context of system factorization and the
associated image subspace. Specifically, the matrix $\mathcal{H}_{s+n}\left(
v_{\left[ k_{p}+1:k_{p}+N\right] }\right) $ of the latent vector implies the
hidden (persistent) excitation that enables the system input-output data set 
$T_{G,s}$ to comprise full information of the system image representation $%
\left( M_{s},N_{s}\right) .$ Notice that (\ref{eq3-6})-(\ref{eq3-6a}) mean
\begin{equation}
\setlength{\abovedisplayskip}{0pt}\setlength{\belowdisplayskip}{0pt}\func{Im}%
\left( \hspace{-2pt}\left[ 
\begin{array}{c}
M_{s} \\ 
N_{s}%
\end{array}%
\right] \hspace{-2pt}\mathcal{H}_{s+n}\left( v_{\left[ k_{p}+1:k_{p}+N\right]
}\right) \hspace{-2pt}\right) \hspace{-2pt}=\hspace{-2pt}\func{Im}\left( 
\hspace{-2pt}\left[ 
\begin{array}{c}
M_{s} \\ 
N_{s}%
\end{array}%
\right] \hspace{-2pt}\right) .  \label{eq3-6b}
\end{equation}%
Equation (\ref{eq3-6b}) provides us with a necessary and sufficient
condition for the required excitation. It can be interpreted as thoroughly
exciting the system image subspace. It is noteworthy that, if $\mathcal{H}%
_{s+n}\left( v_{\left[ k_{p}+1:k_{p}+N\right] }\right) $ is persistently
exciting of order $s+n,$ i.e. 
\begin{equation*}
\setlength{\abovedisplayskip}{6pt}\setlength{\belowdisplayskip}{6pt}%
rank\left( \mathcal{H}_{s+n}\left( v_{\left[ k_{p}+1:k_{p}+N\right] }\right)
\right) =\left( s+n\right) p,
\end{equation*}%
the image representation $\left( M_{s},N_{s}\right) $ is identifiable on the
basis of the system input-output data set, namely 
\begin{align*}
\left[ 
\begin{array}{c}
M_{s} \\ 
N_{s}%
\end{array}%
\right] & =T_{G,s}\mathcal{H}_{s+n}^{T}\left( v_{\left[ k_{p}+1:k_{p}+N%
\right] }\right) \\
& \times \left( \mathcal{H}_{s+n}\left( v_{\left[ k_{p}+1:k_{p}+N\right]
}\right) \mathcal{H}_{s+n}^{T}\left( v_{\left[ k_{p}+1:k_{p}+N\right]
}\right) \right) ^{-1}.
\end{align*}%
In fact, in this case, there exists a matrix $\Pi \in \mathbb{R}^{\left(
N-s+1\right) \times \left( N-s+1\right) }$ s.t. 
\begin{equation*}
\setlength{\abovedisplayskip}{6pt}\setlength{\belowdisplayskip}{6pt}\mathcal{%
H}_{s+n}\left( u_{\left[ k_{0}+1:k_{0}+N\right] }\right) =\mathcal{H}%
_{s+n}\left( \left( v_{\left[ k_{p}+1:k_{p}+N\right] }\right) \right) \Pi .
\end{equation*}%
The image representation (\ref{eq3-1}) and the associated equivalence (\ref%
{eq3-6})\textit{\ }provide alternative representation forms of system
dynamics, which build the core of our subsequent work.

\begin{definition}
\label{Def3-1} Given the LTI system (\ref{eq2-10})-(\ref{eq2-11}), the
input-output reprsentations (\ref{eq3-1}) and (\ref{eq3-6a}) are called
finite sample image representation and data-driven image representation,
respectively. 
\end{definition}

\subsection{Residual subspace and extension of fundamental lemma}

The system image representation and also fundamental lemma model the system
nominal dynamics exclusively. On the other hand, addressing system
uncertainties including faults is a necessary step to perform fault
detection. As illustrated in Section 2, in the model-based framework,
residual generation is the first step for a reliable fault detection. The
system kernel representation (\ref{eq2-3}) serves as the core of any
observer-based residual generator, and associated with it, fault detection
is achieved in the residual subspace. The objective of this subsection is
twofold. Firstly, dynamics of LTI systems with uncertainties are modelled in
the space of finite sample signals. In this regard, the fundamental lemma is
extended. Secondly, existence conditions of kernel representation and
residual subspace are examined.

Consider the system model (\ref{eq2-20})-(\ref{eq2-21}) and the associated
kernel-based model (\ref{eq2-25}). Particularly, for systems with the
process and sensor noises as uncertainties modelled by (\ref{eq2-20a})-(\ref%
{eq2-21b}), when $L$ is set as the Kalman filter gain $L_{K}$, 
\begin{align*}
L_{K}& =\left( APC^{T}+S_{\omega v}\right) \Sigma _{r}^{-1},P=APA^{T}+\Sigma
_{\omega }, \\
\Sigma _{r}& =CPC^{T}+\Sigma _{\upsilon }=\mathbb{E}\left(
r(k)r^{T}(k)\right) ,
\end{align*}%
the residual $r(k)\sim \mathcal{N}(0,\Sigma _{r})$ is an innovation series
and the system (\ref{eq2-25}) is known as innovation model %
\citep{Huang_2008_book}. The following lemma is essential for our subsequent
work.

\begin{lemma}
\label{Lemma3-3}Given the kernel-based model (\ref{eq2-25}), it holds, for
some $v_{s+n}(k-n),$ 
\begin{gather}
\left[ 
\begin{array}{c}
u_{s}(k) \\ 
y_{s}(k)%
\end{array}%
\right]\hspace{-2pt} =\hspace{-2pt}\left[ 
\begin{array}{c}
M_{s} \\ 
N_{s}%
\end{array}%
\right] v_{s+n}(k\hspace{-2pt}-\hspace{-2pt}n)\hspace{-2pt}+\hspace{-2pt}\left[ 
\begin{array}{c}
\hat{Y}_{s} \\ 
\hat{X}_{s}%
\end{array}%
\right] \hspace{-2pt}r_{s+n}(k\hspace{-2pt}-\hspace{-2pt}n)  \label{eq3-8} \\
\left[ 
\begin{array}{c}
\hat{Y}_{s} \\ 
\hat{X}_{s}%
\end{array}%
\right] =\left[ 
\begin{array}{cc}
H_{u,\hat{x}} & \text{ }H_{u,r,F_{d}} \\ 
H_{y,\hat{x}} & \text{ }H_{y,r,F_{d}}%
\end{array}%
\right] , \\
H_{u,\hat{x}}=\left[ 
\begin{array}{c}
F_{d} \\ 
\vdots  \\ 
F_{d}A_{F_{d}}^{s-1}%
\end{array}%
\right] \left[ 
\begin{array}{ccc}
A_{F_{d}}^{n-1}L & \cdots  & \text{ }L%
\end{array}%
\right] \in \mathbb{R}^{sp\times nm},  \notag \\
H_{u,r,F_{d}}=\left[ 
\begin{array}{cccc}
0 & 0 & \cdots  & \text{ }0 \\ 
F_{d}L & 0 & \ddots  & \vdots  \\ 
\vdots  & \ddots  & \ddots  & \text{ }0 \\ 
F_{d}A_{F_{d}}^{s-2}L & \cdots  & F_{d}L & \text{ }0%
\end{array}%
\right] \in \mathbb{R}^{sp\times sm},  \notag \\
H_{y,\hat{x}}=\left[ 
\begin{array}{c}
C_{F_{d}} \\ 
\vdots  \\ 
C_{F_{d}}A_{F_{d}}^{s-1}%
\end{array}%
\right] \left[ 
\begin{array}{ccc}
A_{F_{d}}^{n-1}L & \cdots  & \text{ }L%
\end{array}%
\right] \in \mathbb{R}^{sm\times nm},  \notag \\
H_{y,r,F_{d}}=\left[ 
\begin{array}{cccc}
I & 0 & \cdots  & \text{ }0 \\ 
C_{F_{d}}L & \text{ }I & \ddots  & \vdots  \\ 
\vdots  & \ddots  & \ddots  & \text{ }0 \\ 
C_{F_{d}}A_{F_{d}}^{s-2}L & \cdots  & C_{F_{d}}L & \text{ }I%
\end{array}%
\right] \in \mathbb{R}^{sm\times sm}, \notag
\end{gather}%
where $\left( M_{s},N_{s}\right) $ are defined by (\ref{eq3-2}) for any $F,$
and $F_{d}$ is the deadbeat gain matrix, i.e. all eigenvalues of matrix $%
A_{F_{d}}=A+BF_{d}$ are equal to zero. 
\end{lemma}
\begin{proof}
It follows from (\ref{eq2-4}) and the state space
realizations (\ref{eq2-1a}) and (\ref{eq2-2}) that for $u=F_{d}\hat{x}+\bar{v%
}$\ 
\begin{align*}
\hat{x}(k+1)& =A_{F_{d}}\hat{x}(k)+B\bar{v}(k)+Lr(k), \\
\left[ 
\begin{array}{c}
u(k) \\ 
y(k)%
\end{array}%
\right] & =\left[ 
\begin{array}{c}
F_{d} \\ 
C_{F_{d}}%
\end{array}%
\right] \hat{x}(k)+\left[ 
\begin{array}{c}
I \\ 
D%
\end{array}%
\right] v(k)+\left[ 
\begin{array}{c}
0 \\ 
I%
\end{array}%
\right] r(k).
\end{align*}%
Hence, according to Theorem \ref{Theorem 3-1}, the response of $\left(
u_{s}(k),y_{s}(k)\right) $\ to $\bar{v}$ can be described by, for any $F,$\ 
\begin{align*}
\left[ 
\begin{array}{c}
M_{d,s} \\ 
N_{d,s}%
\end{array}%
\right] \bar{v}_{s+n}(k-n)& =\left[ 
\begin{array}{c}
M_{s} \\ 
N_{s}%
\end{array}%
\right] v_{s+n}(k-n), \\
v_{s+n}(k-n)& =V_{s+n}^{-1}\bar{v}_{s+n}(k-n).
\end{align*}%
Analogous it, we obtain $\left[ 
\begin{array}{c}
\hat{Y}_{s} \\ 
\hat{X}_{s}%
\end{array}%
\right] r_{s+n}(k-n)$\ as the response to the residual $r,$\ whose proof can
be performed along the lines of the proof of Lemma \ref{Lemma3-1} given in %
\citep{Ding_automatica2014}. This completes the proof.
\end{proof}

Recall that for two different observer gain matrices, $L=L_{1},L=L_{2}\neq
L_{1},$ the observer-based residual generator $r=\hat{M}y-\hat{N}u$ satisfies%
\begin{align}
r& =\hat{M}_{2}y-\hat{N}_{2}u=R\left( \hat{M}_{1}y-\hat{N}_{1}u\right) ,
\label{eq3-23} \\
R& =\left( A_{L_{2}},L_{1}-L_{2},C,I\right) ,  \label{eq3-24}
\end{align}%
where $R$ is an invertible postfilter \citep{Ding2008}. It is of theoretic
interest to examine the change of $\left( \hat{X}_{s},\hat{Y}_{s}\right) $
caused by variations of $L$ and $F.$ The results are described in the
following lemma.

\begin{lemma}
\label{Lemma3-4}Given the kernel-based model (\ref{eq2-25}), and suppose
that $(\hat{X}_{s,i},\hat{Y}_{s,i})$ and $\left( M_{s,i},N_{s,i}\right) $
are constructed according to (\ref{eq3-8}) and (\ref{eq3-1a}), respectively,
for $L=L_{i},F=F_{i},i=1,2,$ Then, it holds 
\begin{gather}
\begin{bmatrix}
\hat{Y}_{s,1}\, \\ 
\hat{X}_{s,1}\,%
\end{bmatrix}%
=\left[ 
\begin{array}{cc}
M_{s,2} & \hat{Y}_{s,2} \\ 
N_{s,2} & \hat{X}_{s,2}%
\end{array}%
\right] \left[ 
\begin{array}{c}
\bar{R}_{s+n} \\ 
R_{s+n}%
\end{array}%
\right] ,  \label{eq3-50} \\
R_{s+n}=\mathcal{T}_{s+n,s+n}\left( R\right) ,R_{k}=\left\{ 
\begin{array}{l}
CA_{L_{2}}^{k-1}(L_{1}-L_{2}),k\geq 1 \\ 
I,k=0 \\ 
0,k<0,%
\end{array}%
\right.   \notag \\
\bar{R}_{s+n}=\bar{R}_{s+n}^{1}+\bar{R}_{s+n}^{2}=\mathcal{T}%
_{s+n,s+n}\left( \bar{R}\right) ,\bar{R}=\bar{R}^{1}+\bar{R}^{2},  \notag \\
\bar{R}_{k}^{1}=\left\{ 
\begin{array}{l}
(F_{1}-F_{2})A_{F_{1}}^{k-1}L_{1},k\geq 1 \\ 
0,k\leq 0,%
\end{array}%
\right.  \\
\bar{R}_{k}^{2}=\left\{ 
\begin{array}{l}
F_{2}A_{L_{2}}^{k-1}(L_{1}-L_{2}),k\geq 1 \\ 
0,k\leq 0,%
\end{array}%
\right.  \\
\left( \mathcal{T}_{s+n,s+n}\left( R\right) \right) _{i,j}=R_{i-j},\left( 
\mathcal{T}_{s+n,s+n}\left( \bar R\right) \right) _{i,j}=\bar{R}_{i-j}.  \notag
\end{gather}
\end{lemma}
\begin{proof}
The proof is performed in two steps, (i) varying $%
L $\ with a fixed $F,$\ (ii) fixing $L,$\ varying $F.$\ The identity 
\begin{equation}
A_{F}^{k}-A_{L}^{k}=\dsum\limits_{i=0}^{k-1}A_{F}^{k-1-i}\left( BF+LC\right)
A_{L}^{i},  \label{eq5-0}
\end{equation}%
plays a key role in our proof, which can be easily proven by induction
method and noticing $BF+LC=A_{F}-A_{L}.$\ We begin with the first step and
introduce notation%
\begin{gather*}
\hat{Y}_{s,i}=\mathcal{T}_{s,s+n}\left( Y^{i}\right) ,Y_{k}^{i}=\left\{ 
\begin{array}{l}
FA_{F}^{k-1}L_{i},k\geq 1 \\ 
0,k\leq 0,%
\end{array}%
\right. \\
\hat{X}_{s,i}=\mathcal{T}_{s,s+n}\left( X^{i}\right) ,X_{k}^{i}=\left\{ 
\begin{array}{l}
C_{F}A_{F}^{k-1}L_{i},k\geq 1 \\ 
I,k=0 \\ 
0,k<0,%
\end{array}%
\right. \\
\left( \mathcal{T}_{s,s+n}\left( Y^{i}\right) \right)
_{l,j}=Y_{n+l-j}^{i},\left( \mathcal{T}_{s,s+n}\left( X^{i}\right) \right)
_{l,j}=X_{n+l-j}^{i}.
\end{gather*}%
To be proven is $\hat{Y}_{s,2}R_{s+n}=\hat{Y}_{s,1}-M_{s,2}\bar{R}%
_{s+n}^{2}, $\ for $L=L_{i},i=1,2,$\ and fixing $F$ $=F_{2}.$\ Consider the
Toeplitz matrix of $\hat{Y}_{s,2}R_{s+n},$\ which, by routine computation
according to Lemma \ref{Lemma3-3} and definition of $R_{s+n}$, is given by 
\begin{gather*}
\hat{Y}_{s,2}R_{s+n}=\mathcal{T}_{s,s+n}\left( \bar{Y}\right) ,\left( 
\mathcal{T}_{s,s+n}\left( \bar{Y}\right) \right) _{l,j}=\bar{Y}_{n+l-j}, \\
\bar{Y}_{k}=\left\{ 
\begin{array}{l}
F_{2}\left( 
\begin{array}{c}
\dsum\limits_{l=0}^{k-2}A_{F_{2}}^{l}L_{2}CA_{L_{2}}^{k-2-l}(L_{1}-L_{2}) \\ 
+A_{F_{2}}^{k-1}L_{2}%
\end{array}%
\right) ,k\geq 1 \\ 
0,k\leq 0.%
\end{array}%
\right.
\end{gather*}%
\ Attributed to (\ref{eq5-0}), 
\begin{gather}
\dsum\limits_{l=0}^{k-2}A_{F_{2}}^{l}L_{2}CA_{L_{2}}^{k-2-l}(L_{1}-L_{2})= 
\notag \\
\left(
A_{F_{2}}^{k-1}-A_{L_{2}}^{k-1}-\dsum%
\limits_{l=0}^{k-2}A_{F_{2}}^{l}BF_{2}A_{L_{2}}^{k-2-l}\right) \left(
L_{1}-L_{2}\right)  \notag \\
\Longrightarrow F_{2}\left(
A_{F_{2}}^{k-1}L_{2}+\dsum%
\limits_{l=0}^{k-2}A_{F_{2}}^{l}L_{2}CA_{L_{2}}^{k-2-l}(L_{1}-L_{2})\right) =
\notag \\
F_{2}\left( 
\begin{array}{c}
-\left(
A_{L_{2}}^{k-1}+\dsum\limits_{l=0}^{k-2}A_{F_{2}}^{l}BF_{2}A_{L_{2}}^{k-2-l}%
\right) (L_{1}-L_{2}) \\ 
+A_{F_{2}}^{k-1}L_{1}%
\end{array}%
\right) .  \label{eq5-1}
\end{gather}%
Note that, according to $M_{s,2}$ and $\bar{R}^{2}$ given in (\ref{eq5-2})
and Lemma \ref{Lemma3-4}, respectively, 
\begin{gather}
\mathcal{T}_{s,s+n}\left( \Xi \right) =M_{s,2}\bar{R}_{s+n}^{2},\left( 
\mathcal{T}_{s,s+n}\left( \Xi \right) \right) _{l,j}=\Xi _{n+l-j},  \notag \\
\Xi _{k}=\left\{ 
\begin{array}{l}
F_{2}\left( 
\begin{array}{c}
\dsum\limits_{t=0}^{k-2}A_{F_{2}}^{t}BF_{2}A_{L_{2}}^{k-2-t} \\ 
+A_{L_{2}}^{k-1}%
\end{array}%
\right) (L_{1}-L_{2}),k\geq 1 \\ 
0,k\leq 0.%
\end{array}%
\right.  \label{eq5-1a}
\end{gather}%
Summarizing (\ref{eq5-1}) and (\ref{eq5-1a}) results in 
\begin{equation*}
\mathcal{T}_{s,s+n}\left( \bar{Y}\right) =\mathcal{T}_{s,s+n}\left( \Xi
\right) \Longrightarrow \hat{Y}_{s,1}=\hat{Y}_{s,2}R_{s+n}+M_{s,2}\bar{R}%
_{s+n}^{2}.
\end{equation*}%
Similarly, it turns out that%
\begin{gather*}
\hat{X}_{s,2}R_{s+n}=\mathcal{T}_{s,s+n}\left( \bar{X}\right) ,\left( 
\mathcal{T}_{s,s+n}\left( \bar{X}\right) \right) _{l,j}=\bar{X}_{n+l-j}, \\
\bar{X}_{k}=\left\{ 
\begin{array}{l}
C_{F_{2}}\dsum%
\limits_{t=0}^{k-2}A_{F_{2}}^{t}L_{2}CA_{L_{2}}^{k-2-t}(L_{1}-L_{2}) \\ 
+C_{F_{2}}A_{F_{2}}^{k-1}L_{2}+CA_{L_{2}}^{k-1}(L_{1}-L_{2}),k>0 \\ 
I,k=0 \\ 
0,k\leq 0.%
\end{array}%
\right.
\end{gather*}%
which, according to (\ref{eq5-0}), equals\ 
\begin{equation*}
\bar{X}_{k}=\left\{ 
\begin{array}{l}
-C_{F_{2}}\left(
A_{L_{2}}^{k-1}+\dsum\limits_{t=0}^{k-2}A_{F_{2}}^{t}BF_{2}A_{L_{2}}^{k-2-t}%
\right) (L_{1}-L_{2}) \\ 
+C_{F_{2}}A_{F_{2}}^{k-1}L_{1}+CA_{L_{2}}^{k-1}(L_{1}-L_{2}),k>0 \\ 
I,k=0 \\ 
0,k\leq 0.%
\end{array}%
\right.
\end{equation*}%
Observe that 
\begin{gather*}
\mathcal{T}_{s,s+n}\left( \bar{\Xi}\right) =N_{s,2}\bar{R}_{s+n}^{2},\left( 
\mathcal{T}_{s,s+n}\left( \bar{\Xi}\right) \right) _{l,j}=\bar{\Xi}_{n+l-j}
\\
\bar{\Xi}_{k}=\left\{ 
\begin{array}{l}
\left( 
\begin{array}{c}
C_{F_{2}}\dsum\limits_{t=0}^{k-2}A_{F_{2}}^{t}BF_{2}A_{L_{2}}^{k-2-t} \\ 
+DF_{2}A_{L_{2}}^{k-1}%
\end{array}%
\right) (L_{1}-L_{2}),k>0 \\ 
0,k\leq 0.%
\end{array}%
\right.
\end{gather*}%
Hence, it holds%
\begin{gather*}
\mathcal{T}_{s,s+n}\left( \bar{X}\right) +\mathcal{T}_{s,s+n}\left( \bar{\Xi}%
\right) =\mathcal{T}_{s,s+n}\left( X^{1}\right) \\
\Longrightarrow \hat{X}_{s,1}=\hat{X}_{s,2}R_{s+n}+N_{s}\bar{R}_{s+n}^{2}.
\end{gather*}%
Next, in the second step, on the assumption of varying $F,F=F_{i},i=1,2,$\
and fixing $L=L_{1},$\ we prove 
\begin{equation*}
\hat{Y}_{s,1}-\hat{Y}_{s,2}=M_{s,2}\bar{R}_{s+n}^{1}\Longleftrightarrow \hat{%
Y}_{s,2}+M_{s,2}\bar{R}_{s+n}^{1}=\hat{Y}_{s,1}.
\end{equation*}%
Note that 
\begin{gather*}
\left( \mathcal{T}_{s,s+n}\left( Y^{1}\right) \right) _{l,j}-\left( \mathcal{%
T}_{s,s+n}\left( Y^{2}\right) \right) _{l,j}= \\
\left\{ 
\begin{array}{l}
\left( F_{1}A_{F_{1}}^{n-1+l-j}-F_{2}A_{F_{2}}^{n-1+l-j}\right) L_{1},n-1+l>j
\\ 
0,n-1+l\leq j,%
\end{array}%
\right.
\end{gather*}%
which is equal to, attributed to the equality (\ref{eq3-5}),%
\begin{gather*}
\left( F_{1}A_{F_{1}}^{n-1+l-j}-F_{2}A_{F_{2}}^{n-1+l-j}\right) L_{1}= \\
\left( \left( F_{1}-F_{2}\right) A_{F_{1}}^{n-1+l-j}+F_{2}\left(
A_{F_{1}}^{n-1+l-j}-A_{F_{2}}^{n-1+l-j}\right) \right) L_{1} \\
=\left( 
\begin{array}{c}
\left( F_{1}-F_{2}\right) A_{F_{1}}^{n-1+l-j} \\ 
+F_{2}\dsum\limits_{t=0}^{n-2+l-j}A_{F_{2}}^{t}B\left( F_{1}-F_{2}\right)
A_{F_{1}}^{n-2+l-j-t}%
\end{array}%
\right) L_{1}
\end{gather*}%
for $n-1+l>j.$ The last equation is the $\left( l,j\right) $-th entry of
Toeplitz matrix $\mathcal{T}_{s,s+n}\left( M_{s,2}\bar{R}_{s+n}^{1}\right) .$%
\ Thus, it holds $\hat{Y}_{s,1}-\hat{Y}_{s,2}=M_{s,2}\bar{R}_{s+n}^{1}.$
Analogously, 
\begin{gather*}
\left( \mathcal{T}_{s,s+n}\left( X^{1}\right) \right) _{l,j}-\left( \mathcal{%
T}_{s,s+n}\left( X^{2}\right) \right) _{l,j}= \\
\left\{ 
\begin{array}{l}
\left( C_{F_{1}}A_{F_{1}}^{n-1+l-j}-C_{F_{2}}A_{F_{2}}^{n-1+l-j}\right)
L_{1},n-1+l>j \\ 
I,n-1+l=j \\ 
0,n-1+l<j,%
\end{array}%
\right.
\end{gather*}%
leads to, attributed to the equality (\ref{eq3-5}),%
\begin{gather*}
\left( C_{F_{1}}A_{F_{1}}^{n-1+l-j}-C_{F_{2}}A_{F_{2}}^{n-1+l-j}\right)
L_{1}= \\
\left( D\left( F_{1}-F_{2}\right) A_{F_{1}}^{n-1+l-j}+C_{F_{2}}\left(
A_{F_{1}}^{n-1+l-j}-A_{F_{2}}^{n-1+l-j}\right) \right) L_{1} \\
=\left( 
\begin{array}{c}
D\left( F_{1}-F_{2}\right) A_{F_{1}}^{n-1+l-j} \\ 
+C_{F_{2}}\dsum\limits_{t=0}^{n-2+l-j}A_{F_{2}}^{t}B\left(
F_{1}-F_{2}\right) A_{F_{1}}^{n-2+l-j-t}%
\end{array}%
\right) L_{1}
\end{gather*}%
for $n-1+l>j.$ The last equation is the $\left( l,j\right) $-th entry of
Toeplitz matrix $\mathcal{T}_{s,s+n}\left( N_{s,2}\bar{R}_{s+n}^{1}\right) ,$%
\ which implies $\hat{X}_{s,1}-\hat{X}_{s,2}=N_{s,2}\bar{R}_{s+n}^{1}.$ With
the successful proof in both steps, the proof of the lemma is completed.
\end{proof}

\begin{theorem}
\label{Theo3-2}Given the kernel-based model (\ref{eq2-25}) and $\left(
M_{s},N_{s}\right) ,\left( \hat{X}_{s},\hat{Y}_{s}\right) $ for arbitrary $F$
and $L,$ the following statements hold: (i) $\left( M_{s},N_{s}\right) $ and 
$\left( \hat{X}_{s},\hat{Y}_{s}\right) $ are parameterized by $%
V_{s+n},R_{s+n},\bar{R}_{s+n}$ given in Lemma \ref{Lemma3-2} and Lemma \ref%
{Lemma3-4}, for $\left( F_{2},L_{2}\right) =\left( 0,0\right) $ and $\left(
F_{1},L_{1}\right) =\left( F,L\right) ,$ respectively, as follows, 
\begin{gather}
\Psi _{s}=\left[ 
\begin{array}{cc}
M_{s} & \hat{Y}_{s} \\ 
N_{s} & \hat{X}_{s}%
\end{array}%
\right] =\Psi _{s,0}\left[ 
\begin{array}{cc}
V_{s+n} & \bar{R}_{s+n} \\ 
0 & R_{s+n}%
\end{array}%
\right] ,  \label{eq3-55} \\
\Psi _{s,0}=\left[ 
\begin{array}{cc}
M_{s,0} & \hat{Y}_{s,0} \\ 
N_{s,0} & \hat{X}_{s,0}%
\end{array}%
\right] =\left[ 
\begin{array}{cc}
I_{G_{0}} & I_{C_{0}}%
\end{array}%
\right] ,  \label{eq3-57} \\
I_{G_{0}}=\left[ 
\begin{array}{cc}
0 & I \\ 
\mathcal{O}_{s}\mathcal{C}_{n} & \text{ }\mathcal{T}_{s,s}\left( G\right) 
\end{array}%
\right] ,I_{C_{0}}=\left[ 
\begin{array}{cc}
0 & 0 \\ 
\text{ }0 & \text{ }I%
\end{array}%
\right] , \\
R_{s+n}=\mathcal{T}_{s+n,s+n}\left( R\right) ,\bar{R}_{s+n}=\bar{R}%
_{s+n}^{1},
\end{gather}%
where $\mathcal{T}_{s,s}\left( G\right) $ is given in (\ref{eq1-1})-(\ref%
{eq1-2}); (ii) any system input-ouput pair $\left( u_{s}(k),y_{s}(k)\right) $
can be expressed by 
\begin{equation}
\left[ 
\begin{array}{c}
u_{s}(k) \\ 
y_{s}(k)%
\end{array}%
\right] =\left[ 
\begin{array}{cc}
M_{s} & \hat{Y}_{s} \\ 
N_{s} & \hat{X}_{s}%
\end{array}%
\right] \left[ 
\begin{array}{c}
v_{s+n}(k-n) \\ 
r_{s+n}(k-n)%
\end{array}%
\right]   \label{eq3-52}
\end{equation}%
for some $\left( v_{s+n}(k-n),r_{s+n}(k-n)\right) ,$ which is equal to 
\begin{gather}
\left[ 
\begin{array}{c}
u_{s}(k) \\ 
y_{s}(k)%
\end{array}%
\right] =\Psi _{s,0}\left[ 
\begin{array}{c}
\bar{v}_{s+n}(k-n) \\ 
\bar{r}_{s+n}(k-n)%
\end{array}%
\right]   \label{eq3-56} \\
=\bar{\Psi}_{s,0}\left[ 
\begin{array}{c}
\bar{v}_{s+n}(k-n) \\ 
\bar{r}_{s}(k)%
\end{array}%
\right] ,\bar{\Psi}_{s,0}=\left[ 
\begin{array}{cc}
M_{s,0} & 0 \\ 
N_{s,0} & I%
\end{array}%
\right] ,  \label{eq3-56a} \\
\bar{v}_{s+n}(k-n)=V_{s+n}v_{s+n}(k-n)+\bar{R}_{s+n}r_{s+n}(k-n),  \notag \\
\bar{r}_{s+n}(k-n)=R_{s+n}r_{s+n}(k-n);  \label{eq3-56b}
\end{gather}%
(iii) the following rank conditions hold%
\begin{equation}
rank\left( \Psi _{s}\right) =rank\left( \Psi _{s,0}\right) =rank\left( \bar{%
\Psi}_{s,0}\right) =s(p+m).  \label{eq3-51}
\end{equation}
\end{theorem}

\begin{proof}
It follows from Lemmas \ref{Lemma3-2} and \ref{Lemma3-4} that 
$\left( M_{s},N_{s}\right) $ given in Lemma \ref{Lemma3-1} hold
for any $F,$ and for any $(F,L),$ $\left( \hat{X}_{s},%
\hat{Y}_{s}\right) $  can be expressed by (\ref{eq3-50}).
Now, let $\left( F_{1},L_{1}\right) =\left( F,L\right) $  and $%
\left( F_{2},L_{2}\right) =\left( 0,0\right) $ be the feedback gain
and observer gain matrices of $\left( M_{s,i},N_{s,i}\right) $  and 
$\left( \hat{X}_{s,i},\hat{Y}_{s,i}\right) $  in Lemmas \ref%
{Lemma3-2} and \ref{Lemma3-4}, respectively. It is straightforward that (\ref%
{eq3-55}) holds and $\left( \hat{X}_{s,2},\hat{Y}_{s,2}\right) =\left( \hat{%
X}_{s,0},\hat{Y}_{s,0}\right) $  and $\left( M_{s,2},N_{s,2}\right)
=\left( M_{s,0},N_{s,0}\right) $  satisfy (\ref{eq3-57}). Statement
(ii) follows immediately from (\ref{eq3-55}) and (\ref{eq3-57}). Since %
\begin{equation}
\setlength{\abovedisplayskip}{6pt}\setlength{\belowdisplayskip}{6pt}
rank\left[ 
\begin{array}{cc}
V_{s+n} & \bar{R}_{s+n} \\ 
0 & R_{s+n}%
\end{array}%
\right] =\left( m+p\right) \left( s+n\right) ,  \label{eq3-54}
\end{equation}%
it is apparent that (\ref{eq3-51}) is true. 
\end{proof}

By virtue of (\ref{eq3-55}), variations of $\left( M_{s},N_{s}\right) $ and $%
\left( \hat{X}_{s},\hat{Y}_{s}\right) $ resulted from varying $\left(
F,L\right) $\textit{\ }are fully parameterized by\textit{\ }$\left( R_{s+n},%
\bar{R}_{s+n},V_{s+n}\right) .$ More importantly, $\Psi _{s,0}$ and
equivalently $\bar{\Psi}_{s,0}$ characterize the elemental system structure
that is invariant with respect to $\left( F,L\right) $. The rank condition (%
\ref{eq3-51}) is the consequence of this property, which plays a core role
in our subsequent work.

An immediate result of Theorem \ref{Theo3-2} is an extension of the
fundamental lemma.

\begin{corollary}
\label{Co3-2} Consider the kernel-based model (\ref{eq2-25}). For an
input-output data matrix $T_{G,s}$ with $N-s+1\geq \left( s+n\right) p+sm,$
(i) there exists a matrix of the latent variables $\left( \bar{v}_{s+n}(k-n),%
\bar{r}_{s}(k)\right) $ so that for $k_{p}=k-n$%
\begin{equation}
\setlength{\abovedisplayskip}{6pt}\setlength{\belowdisplayskip}{6pt}
\func{Im}\left( T_{G,s}\right) =\func{Im}\left( \bar{\Psi}_{s,0}\left[ 
\begin{array}{c}
\mathcal{H}_{s+n}\left( \bar{v}_{\left[ k_{p}+1:k_{p}+N\right] }\right)  \\ 
\mathcal{H}_{s}\left( \bar{r}_{\left[ k+1:k+N\right] }\right) 
\end{array}%
\right] \right)   \label{eq3-40}
\end{equation}%
and $\forall g\in \mathbb{R}^{N-s+1},\left[ 
\begin{array}{c}
u_{s}(k) \\ 
y_{s}(k)%
\end{array}%
\right] =T_{G,s}g$ is an input-output trajectory of the system; (ii)%
\begin{gather*}
rank\left[ 
\begin{array}{c}
\mathcal{H}_{s+n}\left( \bar{v}_{\left[ k_{p}+1:k_{p}+N\right] }\right)  \\ 
\mathcal{H}_{s}\left( \bar{r}_{\left[ k+1:k+N\right] }\right) 
\end{array}%
\right] =\left( s+n\right) p+sm\Longrightarrow  \\
rank\left( T_{G,s}\right) =s(p+m)\Longleftrightarrow \func{Im}\left(
T_{G,s}\right) =\func{Im}\left( \bar{\Psi}_{s,0}\right) .
\end{gather*}
\end{corollary}

\begin{proof}
Equation (\ref{eq3-40}) follows from Theorem \ref{Theo3-2}
immediately, based on which and the rank condition (\ref{eq3-51}),
conclusion (ii) holds.
\end{proof}

Since $T_{G,s}\in R^{s(p+m)\times \left( N-s+1\right) },$ the statement (ii)
implies that $\left( u_{s}(k),y_{s}(k)\right) $ is an input-output
trajectory of the system if and only if $\exists g\in \mathbb{R}^{N-s+1}$ so
that $\left[ 
\begin{array}{c}
u_{s}(k) \\ 
y_{s}(k)%
\end{array}%
\right] =T_{G,s}g.$ We would like to remark that in %
\citep{WANGautomatica2025}\textit{\ }systems modelled by (\ref{eq2-25}) with
the innovation series $r(k)$ are considered, in which the input vector is
extended to $\left[ 
\begin{array}{cc}
u_{s}^{T}(k) & \text{ }r_{s}^{T}(k)%
\end{array}%
\right] ^{T}$. In this regard, it has been proven (Theorem 2 (a), %
\citep{WANGautomatica2025})\textit{\ }that the data set consisting of $%
\left( u_{s}(k),y_{s}(k),r_{s}(k)\right) $ fully characterizes the system
dynamics if the augmented input vector is persistently exciting of order $%
s+n.$ It is straightforward to showcase that this result is equivalent to
(ii) of Corollary \ref{Co3-2}.

The rank condition (\ref{eq3-40})\textit{\ }illustrates that under existence
of $r$ representing uncertainties in the system, (i) the dimension of the
image subspace of the input-output data, $\func{Im}\left( T_{G,s}\right) ,$
can be higher than $sp+n,$ which implies the existence of the complementary
subspace of the system image subspace, (ii) there exists mutual coupling
between these two subspaces. Both issues are of importance in dealing with
fault detection and thus motivate us to study the kernel representation and
residual subspace.

In the model-based framework, the residual subspace is the complementary
subspace of the image subspace, which, according to (\ref{eq3-13a}), is the
image subspace spanned by the RCP $\left( \hat{X},\hat{Y}\right) $ as well.
Consequently, the image and residual subspaces are characterized by the
properties that $\dim \left( \mathcal{R}_{G}\right) =m,\dim \left( \mathcal{I%
}_{G}\right) =p,$ 
\begin{equation}
\setlength{\abovedisplayskip}{6pt}\setlength{\belowdisplayskip}{6pt} \left[ 
\begin{array}{cc}
-\hat{N} & \hat{M}%
\end{array}%
\right] \left[ 
\begin{array}{cc}
M & \text{ }-\hat{Y} \\ 
N & \text{ }\hat{X}%
\end{array}%
\right] =\left[ 
\begin{array}{cc}
0 & \text{ }I%
\end{array}%
\right] .  \label{eq3-61}
\end{equation}%
This is the theoretic basis for a kernel representation serving as a
residual generator. Analogous to it, to detect faults in the subspace of
finite-sample input-output signals, it is necessary that 
\begin{gather}
\func{Im}\left( \left[ 
\begin{array}{c}
-\hat{Y}_{s} \\ 
\hat{X}_{s}%
\end{array}%
\right] \right) \nsubseteq \func{Im}\left( \left[ 
\begin{array}{c}
M_{s} \\ 
N_{s}%
\end{array}%
\right] \right) ,  \label{eq3-18a} \\
rank\left( \left[ 
\begin{array}{c}
M_{s} \\ 
N_{s}%
\end{array}%
\right] \right) <s\left( p+m\right) .  \label{eq3-18b}
\end{gather}%
Below, we examine the existence conditions for (\ref{eq3-18a})-(\ref{eq3-18b}%
), which are equivalent to the existence conditions of the complementary
subspace of $\mathcal{I}_{G,s}$. As delineated by the subsequent theorem, in
the space of finite-sample signals, kernel representation and residual
subspace exist only under certain conditions.

\textbf{Remark 3}. For the sake of simplicity, notation%
\begin{equation*}
\setlength{\abovedisplayskip}{6pt}\setlength{\belowdisplayskip}{6pt} I_{G,s}=%
\left[ 
\begin{array}{c}
M_{s} \\ 
N_{s}%
\end{array}%
\right] ,I_{C,s}=\left[ 
\begin{array}{c}
\hat{Y}_{s} \\ 
\hat{X}_{s}%
\end{array}%
\right]
\end{equation*}%
is introduced in the sequel.

\begin{theorem}
\label{Theorem3-3}Given the kernel-based model (\ref{eq2-25}) and the
corresponding input-output model (\ref{eq3-8}), then (i) 
\begin{equation}
rank\left( I_{G,s}\right) =\beta +sp,\beta =rank\left( \mathcal{O}%
_{s}\right) ;  \label{eq3-10a}
\end{equation}%
(ii) the rank condition 
\begin{equation}
rank\left( I_{G,s}\right) <s\left( m+p\right)   \label{eq3-15}
\end{equation}%
holds if and only if $sm>\beta ,$ and (iii) if $s>\mu _{obs},$ where $\mu
_{obs}\left( \leq n\right) $ is the system observability index, it holds that%
\begin{gather}
rank\left( I_{G,s}\right) =sp+n<s\left( p+m\right) ,  \label{eq3-17} \\
\dim \left( \mathcal{R}_{G,s}\right) =s(p+m)-\dim \left( \mathcal{I}%
_{G,s}\right) =sm-n,  \label{eq3-19} \\
\dim \left( \mathcal{R}_{G,s}\right) <\dim \left( \func{Im}\left(
I_{C,s}\right) \right) =sm,  \label{eq3-20a}
\end{gather}%
where residual subspace $\mathcal{R}_{G,s}$ is the complementary subspace of 
$\mathcal{I}_{G,s},$ and there exist $\left( u_{s}(k),y_{s}(k)\right) \left(
\neq 0\right) $ so that 
\begin{equation}
\left[ 
\begin{array}{c}
u_{s}(k) \\ 
y_{s}(k)%
\end{array}%
\right] \in \func{Im}\left( I_{G,s}\right) \cap \func{Im}\left(
I_{C,s}\right) .  \label{eq3-21}
\end{equation}
\end{theorem}

\begin{proof}
To prove (\ref{eq3-10a}), setting $F=0$\ leads to%
\begin{equation}
rank\left( I_{G,s}\right) =sp+rank\left( \mathcal{O}_{s}\mathcal{C}%
_{n}\right) .  \label{eq3-13}
\end{equation}%
Since the system (\ref{eq2-10})-(\ref{eq2-11}) is controllable, we have%
\begin{equation*}
rank\left( \mathcal{O}_{s}\mathcal{C}_{n}\right) =\beta \Longrightarrow
rank\left( I_{G,s}\right) =sp+\beta .
\end{equation*}%
The proof of (ii) follows immediately from (\ref{eq3-10a}). Concerning
(iii), recall that\textit{\ }$rank\left( \mathcal{O}_{\mu _{obs}}\right) =n.$
Thus, $\forall s>\mu _{obs},$\ (\ref{eq3-17}) holds, and (\ref{eq3-19}) is
the immediate result of (\ref{eq3-17}). Recalling the rank condition (refer
to Theorem \ref{Theo3-2}) 
\begin{equation*}
rank\left( I_{C,s}\right) =rank\left( I_{C_{0}}\right) =sm,
\end{equation*}%
we further have (\ref{eq3-20a}), from which (\ref{eq3-21}) follows. This
completes the proof.
\end{proof}

\textbf{Remark 4}. \textit{The necessary condition }$s>\mu _{obs}$\textit{\
and the conclusion (\ref{eq3-17}) are known in behavioral systems theory,
where they are expressed in terms of the concept lag of the system instead
of the observability index, in order to deal with LTI systems in a general
form \citep{WILLEMS2005,markovsky2021}. For an LTI with the minimal state
space realization (\ref{eq2-10})-(\ref{eq2-11}), the observability index
equals the lag of the system}.

Theorem \ref{Theorem3-3} reveals that the complementary subspace of $%
\mathcal{I}_{G,s}$ exists, only if the data length $s$ is sufficiently
large. In our subsequent study, it is assumed that $s>n\left( \geq \mu
_{obs}\right) .$ The conditions (\ref{eq3-19})-(\ref{eq3-21}) demonstrate
the differences between the image and residual subspaces in $\mathcal{H}_{2}$
space and in the space of finite-sample signals. They reveal that, in the
latter case, (i) increasing the evaluation window $s$ is instrumental to
enhance detection performance, (ii) $\func{Im}\left( I_{C,s}\right) $ is not
the residual subspace, and thus (iii) the nominal and uncertain system
dynamics are stronger corrupted. Aiming at fault detection, it is essential
to determine a kernel representation as the dual expression of the image
representation $I_{G,s}$ and the basis of residual generation in the
residual subspace of finite sample signals. To this end, we now introduce
the definition of finite-sample kernel representation.

\begin{definition}
\label{def3-2}Given the kernel-based model (\ref{eq2-25}) and the associated
finite sample input-output model (\ref{eq3-8}), system 
\begin{equation}
r_{K }(k)=K_{G,s}\left[ 
\begin{array}{c}
u_{s}(k) \\ 
y_{s}(k)%
\end{array}%
\right] \in \mathbb{R}^{\theta },\theta =sm-n  \label{eq3-19a}
\end{equation}%
is called finite-sample kernel representation, if $K_{G,s}I_{G,s}=0,rank%
\left( K_{G,s}\right) =sm-n.$
\end{definition}

\textbf{Remark 5}. \textit{Hereafter, }$\gamma $\textit{\ and }$\theta $%
\textit{\ are used to denote }%
\begin{align}
\gamma & =rank\left( I_{G,s}\right) =sp+n, \\
\theta & =rank\left( K_{G,s}\right) =sm-n.
\end{align}%
It is evident that the kernel representation (\ref{eq3-19a}) is a residual
generator outputting the residual $r_{K}$ satisfying 
\begin{equation*}
r_{K}(k)=K_{G,s}I_{C,s}r_{s+n}(k-n)\neq 0
\end{equation*}%
as far as $r_{s+n}(k-n)\neq 0.$ Moreover, $\forall R\in \mathbb{R}^{\phi
\times \theta },\phi \leq \theta ,rank\left( R\right) \leq \theta ,$%
\begin{equation}
\bar{r}_{K}(k)=RK_{G,s}I_{C,s}r_{s+n}(k-n)\neq 0,  \label{eq3-21a}
\end{equation}%
which gives a parameterization of all finite-sample residual generators. As
a result, in order to achieve a reliable fault detection, the kernel
representation $K_{G,s}$ (\ref{eq3-19a}) should satisfy the condition 
\begin{equation}
K_{G,s}I_{G,s}=0,rank\left( K_{G,s}I_{C,s}\right) =\theta .  \label{eq3-18}
\end{equation}%
In terms of $K_{G,s}$, the residual subspace is now defined.

\begin{definition}
\label{def3-3}Given the kernel-based model (\ref{eq2-25}) and the associated
finite sample kernel representation (\ref{eq3-19a}), the subspace $\mathcal{R%
}_{G,s},$%
\begin{equation*}
\mathcal{R}_{G,s}=\left\{ 
\begin{array}{c}
\left[ 
\begin{array}{c}
u_{s}(k) \\ 
y_{s}(k)%
\end{array}%
\right] :\left[ 
\begin{array}{c}
u_{s}(k) \\ 
y_{s}(k)%
\end{array}%
\right] =K_{G,s}^{T}r_{K},r_{K}\in \mathbb{R}^{\theta }%
\end{array}%
\right\} 
\end{equation*}%
is called finite-sample residual subspace. 
\end{definition}

\subsection{Kernel representation and parity matrix}

On account of the parameterization of $\Psi _{s}$ described in Theorem \ref%
{Theo3-2},\textit{\ }(\ref{eq3-18}) holds if and only if for $K_{G,s}=\left[ 
\begin{array}{cc}
K_{1} & K_{2}%
\end{array}%
\right] $ 
\begin{gather}
K_{G,s}\Psi _{s,0}=\left[ 
\begin{array}{cccc}
K_{2}\mathcal{O}_{s}\mathcal{C}_{n} & K_{1}+K_{2}\mathcal{T}_{s,s}\left(
G\right) & 0 & K_{2}%
\end{array}%
\right]  \notag \\
=\left[ 
\begin{array}{cccc}
0 & \text{ }0 & \text{ }0 & \text{ }K_{2}%
\end{array}%
\right] \Longrightarrow K_{1}=-K_{2}\mathcal{T}_{s,s}\left( G\right) , 
\notag \\
K_{2}\mathcal{O}_{s}=0,rank\left( K_{2}\right) =sm-n.  \label{eq3-30}
\end{gather}

\begin{theorem}
\label{Theo3-3}Given the kernel-based model (\ref{eq2-25}) and the
finite-sample kernel representation%
\begin{equation}
r_{K}(k)=K_{G,s}\left[ 
\begin{array}{c}
u_{s}(k) \\ 
y_{s}(k)%
\end{array}%
\right] ,  \label{eq3-60}
\end{equation}%
the condition (\ref{eq3-18}) is satisfied if and only if 
\begin{equation}
\setlength{\abovedisplayskip}{6pt}\setlength{\belowdisplayskip}{6pt}%
K_{G,s}=K_{2}\left[ 
\begin{array}{cc}
-\mathcal{T}_{s,s}\left( G\right)  & \text{ }I%
\end{array}%
\right]   \label{eq3-62}
\end{equation}%
with $K_{2}$ satisfying (\ref{eq3-18}). The dynamics of the kernel
representation (\ref{eq3-60})\textit{\ }is governd by 
\begin{equation}
\setlength{\abovedisplayskip}{6pt}\setlength{\belowdisplayskip}{6pt}%
r_{K}(k)=K_{2}\bar{r}_{s}(k).  \label{eq3-38}
\end{equation}
\end{theorem}

\begin{proof}
It follows from (\ref{eq3-30}) that $K_{G,s}$\ subject to (%
\ref{eq3-62})\ is necessary and sufficient for the existence of (\ref{eq3-18}%
). Since 
\begin{equation*}
\setlength{\abovedisplayskip}{6pt}\setlength{\belowdisplayskip}{6pt}%
r_{K}(k)=K_{G,s}\Psi _{s,0}\left[ 
\begin{array}{c}
\bar{v}_{s+n}(k-n) \\ 
\bar{r}_{s+n}(k-n)%
\end{array}%
\right] ,
\end{equation*}%
\ according to Theorem \ref{Theo3-2}, for $K_{G,s}$ subject to (\ref{eq3-62}%
) we have 
\begin{equation*}
K_{G,s}\Psi _{s,0}\left[ 
\begin{array}{c}
\bar{v}_{s+n}(k-n) \\ 
\bar{r}_{s+n}(k-n)%
\end{array}%
\right] =K_{2}\bar{r}_{s}(k).
\end{equation*}%
\ The theorem is proven.
\end{proof}

We would like to call the reader's attention to the condition (\ref{eq3-30})
that is identical to the parity matrix condition (\ref{eq2-15}).
Furthermore, the kernel representation (\ref{eq3-60}) is equivalent to the
parity space residual generator (\ref{eq2-14a}). In fact,\ in %
\citep{Ding_automatica2014} the data-driven kernel representation-based
fault detection was introduced based on the parity space method. The above
results give a proof of the kernel representation in a more general
framework.

\section{A data-driven fault detection scheme}

\subsection{Matrix perturbation theory-based problem set-up and theoretic
basis}

Suppose that system input-output data $\left( u(i),y(i)\right) ,i\in \lbrack
k_{0}+1,k_{0}+N],k_{0}\geq n,s>n,$ are collected from the system (\ref%
{eq2-20})-(\ref{eq2-21}) during fault-free operations and ordered into the
input-output data matrix 
\begin{equation*}
\setlength{\abovedisplayskip}{6pt}\setlength{\belowdisplayskip}{6pt}T_{G,s}=%
\left[ 
\begin{array}{c}
\mathcal{H}_{s}\left( u_{\left[ k_{0}+1:k_{0}+N\right] }\right) \\ 
\mathcal{H}_{s}\left( y_{\left[ k_{0}+1:k_{0}+N\right] }\right)%
\end{array}%
\right] ,
\end{equation*}%
where $u$ is persistently exciting of order $s+n,$ and $N$ is sufficiently
large, satisfying%
\begin{equation*}
\setlength{\abovedisplayskip}{6pt}\setlength{\belowdisplayskip}{6pt}%
N-s+1\geq s\left( p+m\right) \Longleftrightarrow N\geq s\left( p+m\right)
+s-1.
\end{equation*}%
According to Theorem \ref{Theo3-2}, for $k_{p}=k_{0}-n$ 
\begin{equation}
\setlength{\abovedisplayskip}{6pt}\setlength{\belowdisplayskip}{6pt}T_{G,s}=%
\left[ 
\begin{array}{cc}
I_{G,s} & \text{ }I_{C,s}%
\end{array}%
\right] \left[ 
\begin{array}{c}
\mathcal{H}_{s+n}\left( v_{\left[ k_{p}+1:k_{p}+N\right] }\right) \\ 
\mathcal{H}_{s+n}\left( r_{\left[ k_{p}+1:k_{p}+N\right] }\right)%
\end{array}%
\right] .  \label{eq4-1}
\end{equation}%
On a practical assumption that 
\begin{equation}
\left\Vert \mathcal{H}_{s+n}\left( r_{\left[ k_{p}+1:k_{p}+N\right] }\right)
\right\Vert <<\left\Vert \mathcal{H}_{s+n}\left( v_{\left[ k_{p}+1:k_{p}+N%
\right] }\right) \right\Vert  \label{eq4-3}
\end{equation}%
the influence of $r$ is addressed as a perturbation, which enables us to
leverage matrix perturbation theory \citep{Stewart_book_1990} to explore
data-driven fault detection issues. To this end, we introduce orthogonal
projections onto a subspace in $\mathbb{R}^{s(m+p)\times \left( sp+n\right)
} $ and the Davis-Kahan theorem, which play a core role in our work. Let $%
I_{G,s,n},$ 
\begin{equation*}
\setlength{\abovedisplayskip}{6pt}\setlength{\belowdisplayskip}{6pt}%
I_{G,s,n}=I_{G,s}\left( I_{G,s}^{T}I_{G,s}\right) ^{-1/2}\Longrightarrow
I_{G,s,n}^{T}I_{G,s,n}=I,
\end{equation*}%
be the normalized image representation of $I_{G,s}$. The operator $\mathcal{P%
}_{{\mathcal{I}}_{G,s}},\mathcal{P}_{{\mathcal{I}}%
_{G,s}}=I_{G,s,n}I_{G,s,n}^{T},$ defines an orthogonal projection that
projects $\left( u_{s}(k),y_{s}(k)\right) $ onto ${\mathcal{I}}_{G,s},$%
\begin{equation}
\setlength{\abovedisplayskip}{6pt}\setlength{\belowdisplayskip}{6pt}%
\begin{bmatrix}
\hat{u}_{s}(k) \\ 
\hat{y}_{s}(k)%
\end{bmatrix}%
=\mathcal{P}_{{\mathcal{I}}_{G,s}}%
\begin{bmatrix}
u_{s}(k) \\ 
y_{s}(k)%
\end{bmatrix}%
=I_{G,s,n}I_{G,s,n}^{T}%
\begin{bmatrix}
u_{s}(k) \\ 
y_{s}(k)%
\end{bmatrix}%
.  \label{eq4-0}
\end{equation}%
Next, do an SVD of the data matrix 
\begin{equation}
\setlength{\abovedisplayskip}{6pt}\setlength{\belowdisplayskip}{6pt}\frac{1}{%
\sqrt{N}}\left[ 
\begin{array}{c}
\mathcal{H}_{s}\left( u_{\left[ k_{0}+1:k_{0}+N\right] }\right) \\ 
\mathcal{H}_{s}\left( y_{\left[ k_{0}+1:k_{0}+N\right] }\right)%
\end{array}%
\right] =U\Sigma V^{T},  \label{eq4-5}
\end{equation}%
and decompose $U$ and $\Sigma $ into 
\begin{eqnarray*}
U &=&\left[ 
\begin{array}{cc}
U_{1} & U_{2}%
\end{array}%
\right] ,U_{1}\in \mathbb{R}^{s(p+m)\times \left( sp+n\right) }, \\
\Sigma &=&diag\left( \Sigma _{1},\Sigma _{2}\right) ,\Sigma _{1}=diag\left(
\sigma _{1},\cdots \sigma _{sp+n}\right) , \\
\Sigma _{2} &=&diag\left( \sigma _{sp+n+1},\cdots ,\sigma _{s\left(
p+m\right) }\right) ,\sigma _{1}\geq \cdots \geq \sigma _{s\left( p+m\right)
}.
\end{eqnarray*}%
On account of (\ref{eq4-1}) and perturbation assumption (\ref{eq4-3}), it is
nature to consider $U_{1}$ and $U_{2}$ as approximations of the normalized
image representation and its (orthogonal) complement, the normalized kernel
representation, 
\begin{gather}
U_{1}\approx I_{G,s,n},U_{2}=U_{1}^{\bot }\approx I_{G,s,n}^{\bot
}=:K_{G,s,n}^{T},  \notag \\
K_{G,s,n}I_{G,s,n}=0,K_{G,s,n}K_{G,s,n}^{T}=I,  \notag \\
K_{G,s,n}^{T}K_{G,s,n}+I_{G,s,n}I_{G,s,n}^{T}=I.  \label{eq4-24}
\end{gather}%
It is apparent that the approximation errors compromise fault detection
performance. It is of essential importance to specify a tight bound of the
approximation errors. In matrix perturbation theory, the Davis-Kahan theorem
is widely applied to evaluating the influence of perturbations. We summarize
the application of the Davis-Kahan theorem \citep{Stewart_book_1990} to our
problem set-up in the following theorem.

\begin{theorem}
\label{Theo4-1}Consider perturbation problem (\ref{eq4-6}). It holds 
\begin{equation}
\left\Vert I_{G,s,n}I_{G,s,n}^{T}-U_{1}U_{1}^{T}\right\Vert _{2}\leq \frac{%
\left\Vert S_{2}\right\Vert _{2}}{\lambda _{\gamma }\left( S_{1}\right) },
\label{eq4-4}
\end{equation}%
where matrices $S_{1},S_{2}$ are defined by%
\begin{gather}
\frac{1}{N}T_{G,s}T_{G,s}^{T}=S=S_{1}+S_{2},S_{1}=I_{G,s}\Sigma
_{v}I_{G,s}^{T},  \label{eq4-6} \\
S_{2}=I_{G,s}\hat{\Sigma}_{vr}I_{C,s}^{T}+I_{C,s}\hat{\Sigma}%
_{vr}^{T}I_{G,s}^{T}+I_{C,s}\hat{\Sigma}_{r}I_{C,s}^{T},  \notag \\
\Sigma _{v}=\frac{\mathcal{H}_{s+n}\left( v_{\left[ k_{p}+1:k_{p}+N\right]
}\right) \mathcal{H}_{s+n}^{T}\left( v_{\left[ k_{p}+1:k_{p}+N\right]
}\right) }{N},  \notag \\
\hat{\Sigma}_{r}=\frac{\mathcal{H}_{s+n}\left( r_{\left[ k_{p}+1:k_{p}+N%
\right] }\right) \mathcal{H}_{s+n}^{T}\left( r_{\left[ k_{p}+1:k_{p}+N\right]
}\right) }{N},  \notag \\
\hat{\Sigma}_{vr}=\frac{\mathcal{H}_{s+n}\left( v_{\left[ k_{p}+1:k_{p}+N%
\right] }\right) \mathcal{H}_{s+n}^{T}\left( r_{\left[ k_{p}+1:k_{p}+N\right]
}\right) }{N},  \notag
\end{gather}%
and $\lambda _{i}\left( S_{1}\right) \geq \lambda _{i+1}\left( S_{1}\right)
,i=1,\cdots ,\gamma ,\lambda _{i}\left( S_{1}\right) =0,i>\gamma ,$ are the
eigenvalues of $S_{1}.$ 
\end{theorem}

Note that $\mathcal{P}_{\mathcal{I}_{U_{1}}}=U_{1}U_{1}^{T}$ is an operator
of orthogonal projection onto the subspace $\mathcal{I}_{U_{1}}=\func{Im}%
\left( U_{1}\right) .$ Recall that 
\begin{equation}
\left\Vert I_{G,s,n}I_{G,s,n}^{T}-U_{1}U_{1}^{T}\right\Vert _{2}=\left\Vert 
\mathcal{P}_{\mathcal{I}_{U_{1}}}-\mathcal{P}_{\mathcal{I}_{G,s}}\right\Vert
.  \label{eq4-22}
\end{equation}%
Here, $\left\Vert \mathcal{P}_{\mathcal{I}_{U_{1}}}-\mathcal{P}_{\mathcal{I}%
_{G,s}}\right\Vert $ denotes the operator norm. Thus, inequality (\ref{eq4-4}%
) gives a bound of the gap metric $\delta \left( \mathcal{I}_{U_{1}},%
\mathcal{I}_{G,s}\right) =\left\Vert \mathcal{P}_{\mathcal{I}_{U_{1}}}-%
\mathcal{P}_{\mathcal{I}_{G,s}}\right\Vert $ that measures the variation of $%
\mathcal{I}_{G,s}$ caused by the perturbation \citep{Kato_book}. Observe
that \textit{\ }%
\begin{equation*}
\mathcal{P}_{\mathcal{I}_{G,s}}-\mathcal{P}_{\mathcal{I}_{U_{1}}}=\mathcal{I}%
-\mathcal{P}_{\mathcal{I}_{U_{1}}}-\left( \mathcal{I}-\mathcal{I}%
_{G,s}\right) =\mathcal{P}_{\mathcal{I}_{U_{2}}}-\mathcal{P}_{\mathcal{I}%
_{G,s}^{\bot }},
\end{equation*}%
where $\mathcal{P}_{\mathcal{I}_{U_{2}}}$ and $\mathcal{P}_{\mathcal{I}%
_{G,s}^{\bot }}$ are operators of orthogonal projection onto $\mathcal{I}%
_{U_{2}}=\func{Im}\left( U_{2}\right) $ and $\mathcal{I}_{G,s}^{\bot },$
respectively. Hence, 
\begin{equation}
\left\Vert \mathcal{P}_{\mathcal{I}_{U_{2}}}-\mathcal{P}_{\mathcal{I}%
_{G,s}^{\bot }}\right\Vert =\left\Vert \mathcal{P}_{\mathcal{I}_{U_{1}}}-%
\mathcal{P}_{\mathcal{I}_{G,s}}\right\Vert \leq \frac{\left\Vert
S_{2}\right\Vert _{2}}{\lambda _{\gamma }\left( S_{1}\right) }.
\label{eq4-23}
\end{equation}%
Now, we are in the position to formulate the data-driven fault detection to
be addressed in the sequel. To begin with, we introduce the following
definition.

\begin{definition}
\label{Def4-1}Given the data matrix $T_{G,s}$ and the SVD (\ref{eq4-5}), the
orthogonal projections defined by the operators $\mathcal{P}_{\mathcal{I}%
_{U_{1}}}$ and $\mathcal{P}_{\mathcal{I}_{U_{1}}^{\bot }},\mathcal{P}_{%
\mathcal{I}_{U_{1}}^{\bot }}=U_{2}U_{2}^{T}=I-U_{1}U_{1}^{T}=:\mathcal{P}_{%
\mathcal{I}_{U_{2}}},$ are called data-driven orthogonal projections onto
the image subspace $\mathcal{I}_{G,s}$ and residual subspace $\mathcal{I}%
_{G,s}^{\bot }$, respectively. 
\end{definition}

The data-driven fault detection problem will be solved in two steps: (i)
residual generation by means of a data-driven orthogonal projection, (ii)
residual evaluation and threshold setting. Below, we present a fault
detection approach.

\subsection{Orthogonal projection-based fault detection}

As showcased by (\ref{eq4-6}), the data subspace $\mathcal{I}_{U_{1}}$
comprises the data collected during fault-free system operations. Let $%
\mathcal{P}_{\mathcal{I}_{U_{1}}}$ be the operator of an orthogonal
projection onto $\mathcal{I}_{U_{1}},$ the projection-based residual
generator is given by%
\begin{equation}
r_{\mathcal{I}_{U_{1}}}(k)=\left( \mathcal{I}-\mathcal{P}_{\mathcal{I}%
_{U_{1}}}\right) 
\begin{bmatrix}
u_{s}(k) \\ 
y_{s}(k)%
\end{bmatrix}%
=%
\begin{bmatrix}
u_{s}(k) \\ 
y_{s}(k)%
\end{bmatrix}%
-%
\begin{bmatrix}
\hat{u}_{s}(k) \\ 
\hat{y}_{s}(k)%
\end{bmatrix}%
.  \label{eq4-25}
\end{equation}%
Recall (\ref{eq4-24}) and the system (\ref{eq3-8}). It turns out 
\begin{gather*}
r_{\mathcal{I}_{U_{1}}}(k)=U_{2}U_{2}^{T}%
\begin{bmatrix}
u_{s}(k) \\ 
y_{s}(k)%
\end{bmatrix}
\\
\approx K_{G,s,n}^{T}K_{G,s,n}\left[ 
\begin{array}{cc}
I_{G,s} & \text{ }I_{C,s}%
\end{array}%
\right] 
\begin{bmatrix}
v_{s+n}(k-n) \\ 
r_{s+n}(k-n)%
\end{bmatrix}
\\
=K_{G,s,n}^{T}K_{G,s,n}I_{C,s}r_{s+n}(k-n).
\end{gather*}%
Noticing that $\left\Vert r_{\mathcal{I}_{U_{1}}}(k)\right\Vert =\left\Vert
U_{2}^{T}%
\begin{bmatrix}
u_{s}(k) \\ 
y_{s}(k)%
\end{bmatrix}%
\right\Vert ,$ for the detection purpose, it is equivalent to construct a
kernel-based residual generator 
\begin{gather}
r_{U_{2}}(k)=U_{2}^{T}%
\begin{bmatrix}
u_{s}(k) \\ 
y_{s}(k)%
\end{bmatrix}
\label{eq4-26} \\
=K_{G,s,n}I_{C,s}r_{s+n}(k-n)+\Delta _{r_{\mathcal{I}_{U_{2}}}}\in \mathbb{R}%
^{\theta },  \label{eq4-48}
\end{gather}%
where $\Delta _{r_{\mathcal{I}_{U_{2}}}}$ addresses the possible error
caused by $K_{G,s,n}-U_{2}.$ On this basis, the corresponding data-driven
fault detection problem is formulated as: given the residual matrix $R_{U_{2}%
\left[ k_{0}+1:k_{0}+N\right] }=U_{2}^{T}T_{G,s}$ and writing 
\begin{gather}
r_{U_{2}}(k)-\Delta _{r_{\mathcal{I}_{U_{2}}}}=d_{U_{2}}(k)+f_{U_{2}}(k),
\label{eq4-27} \\
d_{U_{2}}(k)=K_{G,s,n}I_{C,s}r_{s+n}(k-n),
\end{gather}%
with $f_{U_{2}}(k)$ representing the influence of the system faults on the
residual vector $r_{U_{2}},$ find a residual evaluation function and a
threshold setting that deliver an optimal fault detection. We first consider
the case that $r(k)$ represents an $\ell _{2}$-bounded unknown signal.
Impose $\Sigma _{r_{U_{2}}}=R_{U_{2}\left[ k_{0}+1:k_{0}+N\right] }R_{U_{2}%
\left[ k_{0}+1:k_{0}+N\right] }^{T}.$ On the assumption that\textit{\ }%
\begin{equation*}
rank\left( \mathcal{H}_{s+n}\left( r_{\left[ k_{p}+1:k_{p}+N\right] }\right)
\right) =(s+n)m,
\end{equation*}%
the residual evaluation function is defined by 
\begin{equation*}
J=\bar{r}_{U_{2}}(k)^{T}\Sigma _{r_{U_{2}}}^{-1}\bar{r}_{U_{2}}(k),\bar{r}%
_{U_{2}}(k)=r_{U_{2}}(k)-\hat{\Delta}_{r_{\mathcal{I}_{U_{2}}}}
\end{equation*}%
with $\hat{\Delta}_{r_{\mathcal{I}_{U_{2}}}}$ as an estimate of $\Delta _{r_{%
\mathcal{I}_{U_{2}}}}.$ The threshold setting problem is now formulated as
finding $\hat{\Delta}_{r_{\mathcal{I}_{U_{2}}}}$ and $J_{th}$ so that 
\begin{equation}
\setlength{\abovedisplayskip}{6pt}\setlength{\belowdisplayskip}{6pt}\left\{ 
\begin{array}{l}
J\leq J_{th},\text{ fault-free operations} \\ 
J>J_{th},\text{ faulty.}%
\end{array}%
\right.  \label{eq4-29}
\end{equation}%
The above detection problem is in fact a data-driven realization of the
so-called unified solution \citep{Ding2008}. To solve it, we propose to
apply support vector data description (SVDD) method. SVDD is a
well-established data-driven algorithm, originally developed for solving
one-class classification problems \citep{SVDD2002}. In the SVDD framework,
our problem is formulated as follows. Given the data set, $r_{U_{2}}(k)\in 
\mathbb{R}^{\theta },k=k_{0}+i,i=1,\cdots ,N-s+1,$ find $\hat{\Delta}_{r_{%
\mathcal{I}_{U_{2}}}}\in \mathbb{R}^{\theta },J_{th}\in \mathbb{R}$ and $\xi
_{i}\in \mathbb{R},i=1,\cdots ,N-s+1,$ that solve the minimization problem: 
\begin{gather}
\min_{J_{th},\hat{\Delta}_{r_{\mathcal{I}_{U_{2}}}},\xi _{i}}\left(
J_{th}^{2}+C\sum\limits_{i=1}^{N-s+1}\xi _{i}\right)  \label{eq4-19a} \\
\text{s.t. }\forall \left\Vert \bar{r}_{U_{2}}(k_{0}+i)\right\Vert _{\Sigma
_{r_{U_{2}}}^{-1}}^{2}\leq J_{th}^{2}+\xi _{i},\xi _{i}\geq 0,  \notag
\end{gather}%
where $C\geq 0$ is a design parameter, $i=1,\cdots ,N-s+1,$ and $\xi _{i}$
are slack variables representing tolerable outliers (false alarms). There
exist well-established algorithms to solve the optimization problem (\ref%
{eq4-19a}). The reader is referred to \citep{SVDD2002} for details.

Next, assume that the uncertainties in the system (\ref{eq2-20})-(\ref%
{eq2-21}) are process and measurement noises modelled by (\ref{eq2-20a})-(%
\ref{eq2-21b}), the residual $r(k)\sim \mathcal{N}(0,\Sigma _{r})$ is white
noise series subject to (\ref{eq3-20}). Hence, for sufficiently large $N,$
we have the approximations%
\begin{gather}
\hat{\Sigma}_{r}\approx \Sigma _{r_{s+n}}=\mathbb{E}r_{s+n}\left( k-n\right)
r_{s+n}^{T}\left( k-n\right) ,  \notag \\
\hat{\Sigma}_{vr}\approx \Sigma _{v_{s+n}r_{s+n}}=\mathbb{E}v_{s+n}\left(
k-n\right) r_{s+n}^{T}\left( k-n\right) =0,  \notag \\
\Longrightarrow S_{1}=I_{G,s}\Sigma _{v}I_{G,s}^{T},S_{2}=I_{C,s}\hat{\Sigma}%
_{r}I_{C,s}^{T}.  \label{eq4-2}
\end{gather}%
Consequently, the dynamics of the residual generator (\ref{eq4-26}) is
described by 
\begin{equation*}
\setlength{\abovedisplayskip}{6pt}\setlength{\belowdisplayskip}{6pt}%
r_{U_{2}}(k)=U_{2}^{T}\Psi _{s}\left[ 
\begin{array}{c}
v_{s+n}\left( k-n\right) \\ 
r_{s+n}\left( k-n\right)%
\end{array}%
\right] \sim \mathcal{N}(\Delta _{r_{\mathcal{I}_{U_{2}}}},\Sigma
_{r_{U_{2}}}).
\end{equation*}%
Taking into account of faults, we have 
\begin{equation}
\setlength{\abovedisplayskip}{6pt}\setlength{\belowdisplayskip}{6pt}%
r_{U_{2}}(k)=\Delta _{r_{\mathcal{I}_{U_{2}}}}+\varepsilon
_{U_{2}}(k)+f_{U_{2}}(k)\in \mathbb{R}^{sm-n},  \label{eq4-10}
\end{equation}%
where $\varepsilon _{U_{2}}(k)\sim \mathcal{N}(0,\hat{\Sigma}_{r_{U_{2}}}).$
If $\Delta _{r_{\mathcal{I}_{U_{2}}}}$ is sufficiently small or a
quasi-stationary signal, the detection problem (\ref{eq4-10}) is well solved
using the $\chi ^{2}$-test statistic $J$, 
\begin{gather}
J=\bar{r}_{U_{2}}(k)^{T}\Sigma _{r_{U_{2}}}^{-1}\bar{r}_{U_{2}}(k),
\label{eq4-16} \\
\hat{\Delta}_{r_{\mathcal{I}_{U_{2}}}}=\frac{1}{N}\dsum%
\limits_{k=1}^{N}r_{U_{2}}(k),J_{th}=\chi _{\alpha }^{2}\left( \theta
\right) )  \label{eq4-16a}
\end{gather}%
for a false alarm rate $\alpha $ \citep{Ding2008}. Otherwise, an upper-bound
of $\left\Vert \Delta _{r_{\mathcal{I}_{U_{2}}}}\right\Vert $ is estimated
and the corresponding detection problem is solved by means of the
well-established generalized likelihood ratio (GLR) testing. For details,
the reader is referred to \citep{Basseville93}.

\begin{algorithm}
\textbf{Projection-based data-driven fault detection}
\begin{enumerate}
\item Collect input-output data and build the data matrix $T_{G,s};$

\item Perform the SVD (\ref{eq4-5}) with output $U_{2}$;

\item Construct the residual generator (\ref{eq4-26});

\item Generate $r_{U_{2}}(k)$ and build the data matrices $R_{U_{2}\left[ k_{0}+1:k_{0}+N\right] },\Sigma_{r_{U_{2}}}$; 

\item Run SVDD algorithm to determine $\left( \hat{\Delta}_{r_{\mathcal{I}%
_{U_{2}}}},J_{th}\right) $ or alternatively compute $\left( \hat{\Delta}_{r_{%
\mathcal{I}_{U_{2}}}},J_{th}\right) $ according to (\ref{eq4-16a}).
\end{enumerate}
\end{algorithm}

\subsection{Detection performance analysis}

Performance of a fault detection system is mainly characterized by false
alarm rate (FAR) and miss detection rate (MDR). We now analyze the impact of
the gaps between $\mathcal{I}_{G,s}$ and $\mathcal{I}_{U_{1}}$ as well as
between $\mathcal{R}_{G,s}$ and $\mathcal{I}_{U_{2}}$ on these two
performance indices, which are resulted from the perturbation. As described
in the previous subsection, the perturbation yields 
\begin{align*}
\delta \left( \mathcal{I}_{U_{1}},\mathcal{I}_{G,s}\right) & =\left\Vert 
\mathcal{P}_{\mathcal{I}_{U_{1}}}-\mathcal{P}_{\mathcal{I}_{G,s}}\right\Vert
, \\
\delta \left( \mathcal{I}_{U_{2}},\mathcal{R}_{G,s}\right) & =\left\Vert 
\mathcal{P}_{\mathcal{I}_{U_{2}}}-\mathcal{P}_{\mathcal{I}_{G,s}^{\bot
}}\right\Vert .
\end{align*}%
It is a well-known result that gap metric can be computed by solving a
model-matching problem (MMP) \citep{Kato_book}. Applying it to our case
gives 
\begin{align*}
\delta \left( \mathcal{I}_{U_{1}},\mathcal{I}_{G,s}\right) &
=\min_{Q_{1}}\left\Vert I_{G,s,n}-U_{1}Q_{1}\right\Vert _{2}, \\
\delta \left( \mathcal{I}_{U_{2}},\mathcal{R}_{G,s}\right) &
=\min_{Q_{2}}\left\Vert K_{G,s,n}^{T}-U_{2}Q_{2}\right\Vert _{2}.
\end{align*}%
It is straightforward that 
\begin{gather}
\left\Vert I_{G,s,n}-U_{1}Q_{1}\right\Vert _{2}=\left\Vert \left[ 
\begin{array}{c}
U_{1}^{T} \\ 
U_{2}^{T}%
\end{array}%
\right] \left( I_{G,s,n}-U_{1}Q_{1}\right) \right\Vert _{2}  \notag \\
=\left\Vert \left[ 
\begin{array}{c}
U_{1}^{T}I_{G,s,n}-Q_{1} \\ 
U_{2}^{T}I_{G,s,n}%
\end{array}%
\right] \right\Vert _{2}\Longrightarrow  \notag \\
\delta \left( \mathcal{I}_{U_{1}},\mathcal{I}_{G,s}\right) =\left\Vert
U_{2}^{T}I_{G,s,n}\right\Vert _{2}=\left\Vert
U_{2}U_{2}^{T}I_{G,s,n}\right\Vert _{2},  \label{eq4-7} \\
\left\Vert K_{G,s,n}^{T}-U_{2}Q_{2}\right\Vert _{2}=\left\Vert \left[ 
\begin{array}{c}
U_{1}^{T} \\ 
U_{2}^{T}%
\end{array}%
\right] \left( K_{G,s,n}^{T}-U_{2}Q_{2}\right) \right\Vert _{2}  \notag \\
=\left\Vert \left[ 
\begin{array}{c}
U_{1}^{T}K_{G,s,n}^{T} \\ 
U_{2}^{T}K_{G,s.n}^{T}-Q_{2}%
\end{array}%
\right] \right\Vert _{2}\Longrightarrow  \notag \\
\delta \left( \mathcal{I}_{U_{2}},\mathcal{R}_{G,s}\right) =\left\Vert
U_{1}U_{1}^{T}K_{G,s,n}^{T}\right\Vert _{2}.  \label{eq4-8}
\end{gather}%
Equations (\ref{eq4-7})-(\ref{eq4-8}) mean that the gap between $\mathcal{I}%
_{U_{1}}$and $\mathcal{I}_{G,s}$ is induced by the projection of the system
image subspace onto the data-driven residual subspace, while the gap between 
$\mathcal{I}_{U_{2}}$ and $\mathcal{R}_{G,s}$ is the result of projecting
the system residual subspace onto the data-driven image subspace. In the
context of fault detection, this implies that (i) a larger gap metric
between $\mathcal{I}_{U_{1}}$and $\mathcal{I}_{G,s}$ leads to a higher FAR,
(ii) a larger gap metric between $\mathcal{I}_{U_{2}}$ and $\mathcal{R}%
_{G,s} $ results in a higher MDR. It follows from (\ref{eq4-4}) in Theorem %
\ref{Theo4-1} that the ratio $\left\Vert S_{2}\right\Vert _{2}/\lambda
_{\gamma }\left( S_{1}\right) $ dictates the value of the gap metric.
Consequently, it quantifies the fault detection performance. Noticing that
both $S_{1},S_{2}$ cannot be computed without model knowledge, an empirical
bound on the assumption of (\ref{eq4-3}) is proposed in order to estimate
the possible FAR and MDR. To this end, $\lambda _{\gamma }\left(
S_{1}\right) $ is replaced by $\lambda _{\gamma }\left( S\right) =\sigma
_{\gamma }^{2}\left( \frac{T_{G,s}}{\sqrt{N}}\right) .$ Moreover, observe
that%
\begin{equation*}
\setlength{\abovedisplayskip}{6pt}\setlength{\belowdisplayskip}{6pt}
S_{2}=S-S_{1}\Longleftrightarrow S_{2}U=U\Sigma ^{2}-S_{1}U.
\end{equation*}%
Substituting $S_{1}$ by $U_{1}\Sigma _{1}^{2}U_{1}^{T}$ leads to%
\begin{equation*}
\setlength{\abovedisplayskip}{6pt}\setlength{\belowdisplayskip}{6pt}
S_{2}U\approx U\Sigma ^{2}-U_{1}\Sigma _{1}^{2}\left[ 
\begin{array}{cc}
I & \text{ }0%
\end{array}%
\right] =\left[ 
\begin{array}{cc}
0 & \text{ }U_{2}\Sigma _{2}^{2}%
\end{array}%
\right] .
\end{equation*}%
Hence, $\left\Vert S_{2}\right\Vert _{2}$ is replaced by $\lambda _{\gamma
+1}\left( S\right) ,$ attributed to the relation%
\begin{equation*}
\setlength{\abovedisplayskip}{6pt}\setlength{\belowdisplayskip}{6pt}
\left\Vert S_{2}\right\Vert _{2}\leq \left\Vert S_{2}U\right\Vert
_{2}\approx \lambda _{\max }\left( U_{2}\Sigma _{2}^{2}\right) \leq \lambda
_{\gamma +1}\left( S\right) ,
\end{equation*}%
which finally gives an empirical bound of 
\begin{equation}
\setlength{\abovedisplayskip}{6pt}\setlength{\belowdisplayskip}{6pt} \frac{%
\left\Vert S_{2}\right\Vert _{2}}{\lambda _{\gamma }\left( S_{1}\right) }%
\approx \frac{\lambda _{\gamma +1}\left( S\right) }{\lambda _{\gamma }\left(
S\right) }=\frac{\sigma _{\gamma +1}^{2}\left( \frac{T_{G,s}}{\sqrt{N}}%
\right) }{\sigma _{\gamma }^{2}\left( \frac{T_{G,s}}{\sqrt{N}}\right) }.
\end{equation}

\subsection{A summary and comparison}

At the end of this section, we briefly summarize the proposed data-driven
projection-based detection approach from a substitute aspect and, based on
it, compare it with the existing data-driven detection method %
\citep{Ding_IJP_2014}. The core idea behind the proposed detection approach
is to determine the system image representation $I_{G,s}$ for the given data
matrix $T_{G,s}$ and, aided by it, to generate residual by projecting the
data onto the image/residual subspace. Mathematically, this is a low-rank
approximation problem. Specifically for our application, it is formulated as 
\begin{gather}
\min_{I_{G,s}}\left\Vert \frac{T_{G,s}}{\sqrt{N}}-I_{G,s}\mathcal{H}%
_{s+n}\left( v_{\left[ k_{p}+1:k_{p}+N\right] }\right) \right\Vert _{2}
\label{eq4-40} \\
\text{s.t. }rank\left( I_{G,s}\right) =\gamma =sp+n.  \notag
\end{gather}%
According to Eckart-Young-Mirsky theorem \citep{Eckart_Young_1936}, the SVD (%
\ref{eq4-5}) is the analytical solution of (\ref{eq4-40}), namely the
leading left singular subspace of $\frac{T_{G,s}}{\sqrt{N}}$ gives the
optimal rank-$\gamma $ approximation of $\func{Im}(I_{G,s}).$ Denoted by $%
\left( \hat{u}_{s}(k),\hat{y}_{s}(k)\right) $ the optimal estimate in the
data-driven image subspace, residual $r_{U_{2}}(k)$ is generated based on $%
r_{U_{1}}(k)$ (refer to (\ref{eq4-25})-(\ref{eq4-26})). On the other hand,
the existing data-driven optimal residual generation \citep{Ding_IJP_2014}
follows a different scheme, namely residual generation is based an optimal
output estimation according to (\ref{eq2-9c}). In \citep{Ding_IJP_2014}, by
means of a RQ-decomposition of the input-output data matrix$,$ 
\begin{gather*}
\left[ 
\begin{array}{c}
T_{G,s+\rho } \\ 
\mathcal{H}_{s}\left( y_{\left[ k_{0}+1:k_{0}+N\right] }\right)%
\end{array}%
\right] =\left[ 
\begin{array}{cc}
R_{11} & 0 \\ 
R_{21} & R_{12}%
\end{array}%
\right] \left[ 
\begin{array}{c}
Q_{1} \\ 
Q_{2}%
\end{array}%
\right] , \\
T_{G,s+\rho }=\left[ 
\begin{array}{c}
\mathcal{H}_{\rho }\left( u_{\left[ k_{\rho }+1:k_{\rho }+N\right] }\right)
\\ 
\mathcal{H}_{\rho }\left( y_{\left[ k_{\rho }+1:k_{\rho }+N\right] }\right)
\\ 
\mathcal{H}_{s}\left( u_{\left[ k_{0}+1:k_{0}+N\right] }\right)%
\end{array}%
\right] ,
\end{gather*}%
with $k_{\rho }=k_{0}-\rho ,\rho >n,$ an optimal output estimation is
achieved in the least spares (LS) sense, 
\begin{equation}
\setlength{\abovedisplayskip}{6pt}\setlength{\belowdisplayskip}{6pt}%
\min_{\Phi _{s}}\left\Vert \mathcal{H}_{s}\left( y_{\left[ k_{0}+1:k_{0}+N%
\right] }\right) -\Phi _{s}T_{G,s+\rho }\right\Vert .  \label{eq4-41}
\end{equation}%
In the above equation, 
\begin{equation*}
\setlength{\abovedisplayskip}{6pt}\setlength{\belowdisplayskip}{6pt}R_{y,s}=%
\left[ 
\begin{array}{cc}
-\Phi _{s} & I%
\end{array}%
\right] \left[ 
\begin{array}{c}
T_{G,s+\rho } \\ 
\mathcal{H}_{s}\left( y_{\left[ k_{0}+1:k_{0}+N\right] }\right)%
\end{array}%
\right]
\end{equation*}%
builds the output residual matrix $R_{y,s},$ whose corresponding vectorized
finite-sample model 
\begin{equation*}
\setlength{\abovedisplayskip}{6pt}\setlength{\belowdisplayskip}{6pt}%
r_{y,s}(k)=\left[ 
\begin{array}{cc}
-\Phi _{s} & I%
\end{array}%
\right] \left[ 
\begin{array}{c}
\left[ 
\begin{array}{c}
u_{\rho }(k-\rho ) \\ 
y_{\rho }(k-\rho ) \\ 
u_{s}(k)%
\end{array}%
\right] \\ 
y_{s}(k)%
\end{array}%
\right]
\end{equation*}%
is indeed the finite-sample kernel representation 
\begin{equation}
\setlength{\abovedisplayskip}{6pt}\setlength{\belowdisplayskip}{6pt}%
r_{y,s}(k)=\left[ 
\begin{array}{cc}
-\hat{N}_{s} & \text{ }\hat{M}_{s}%
\end{array}%
\right] \left[ 
\begin{array}{c}
u_{s+\rho }(k-\rho ) \\ 
y_{s+\rho }(k-\rho )%
\end{array}%
\right] ,  \label{eq4-42}
\end{equation}%
a dual form of the image representation (\ref{eq3-1}), which can be derived
from the kernel-based model (\ref{eq2-3}) with the state equation (\ref%
{eq2-9}). It is worth mentioning that an optimal estimation of $R_{y,s}$ in
the sense of (\ref{eq4-41}) was reported in \citep{WANGautomatica2025}.

Comparing (\ref{eq4-40}) and (\ref{eq4-41}) reveals the crucial difference
between these two residual generation schemes. It is the optimal estimation
in two different subspace, one in the input-output data subspace $\left(
u_{s}(k),y_{s}(k)\right) $, while the other one exclusively in the output
subspace $y_{s}(k)$. As illustrated in \citep{DL2026}, a residual in the
former case is capable of handling (multiplicative) uncertainties, including
faults, and thus results in higher fault detectability. In contrast, the
output residual works efficiently to deal with systems with (additive)
unknown inputs. It is noteworthy that 
\begin{equation*}
\dim \left( r_{U_{2}}\right) =(s+\rho )m-n>sm=\dim (r_{y,s}),
\end{equation*}%
a further aspect that illustrates the higher fault detectability of the
projection-based method.

\section{Conclusions}

The objective of this paper is twofold: (i) to introduce the concepts of
image and kernel representations of finite-sample signals of LTI systems
and, associated to them, image and residual subspaces of finite-sample
signals, on this basis, (ii) to develop a data-driven fault detection
approach. Proceeding from the coprime factorizations of LTI systems, the
finite-sample image representation has been derived, which leads to the
definition of image subspace $\mathcal{I}_{G,s}$. The image representation
describes the nominal system behavior, whose data-driven version is
verifiably equivalent to $T_{G,s},$ implying $\mathcal{I}_{G,s}=\func{Im}%
\left( T_{G,s}\right) .$ This study provides an alternative proof of the
fundamental lemma, highlights the role of the latent variable, and showcases
a meaningful aspect behind the excitation condition described in the
fundamental lemma. On the basis of the kernel-based model (\ref{eq2-25}) and
the Bezout identity-induced input-output model (\ref{eq2-4}), a complete
finite-sample model (\ref{eq3-8}) has been established. It describes the
system behavior $\left( u_{s}(k),y_{s}(k)\right) $ under the existence of
uncertainties. The system behavior (\ref{eq3-8}) is fully characterized by
the latent variable pair $\left( v_{s+n},r_{s+n}\right) ,$ and its
data-driven version extends the fundamental lemma. In this regard,
finite-sample kernel representation and residual subspace have been
introduced. It is remarkable that the finite-sample kernel representation
and residual subspace differ from their counterpart induced by the coprime
factorization. Particularly, the statement (\ref{eq3-21}) in Theorem \ref%
{Theorem3-3} underlines the difference between the kernel representation (%
\ref{eq2-9})-(\ref{eq2-9c}) satisfying (\ref{eq3-61}) and the finite-sample
kernel representation (\ref{eq3-60}). It has been proven that the latter is
equivalent to a parity-based residual generator.

Relying on the achieved theoretical results, a data-driven projection-based
fault detection approach has been proposed. This detection approach
comprises three steps, (i) a low-rank matrix approximation of the normalized
image representation $I_{G,s,n}$ by means of an SVD of the data matrix $%
T_{G,s}$, (ii) an orthogonal projection onto the approximated image subspace
and, based on it, generation of projection-based residual vector $%
r_{U_{2}}(k),$ and (iii) threshold setting either by the SVDD algorithm for
the $\ell _{2}$-bounded residual or using $\chi ^{2}$ ($T^{2}$ or GLR)
statistic testing for the residual in the form of an innovation series.
Further effort has been dedicated to the analysis of the detection
performance of the proposed data-driven detection system using the
Davis-Kahan theorem in the framework of matrix perturbation theory. In this
regard, gap metric-based estimation of FAR and MDR has been explored. As a
summary of our work in comparison with the existing data-driven fault
detection methods, the principal difference between the proposed and the
existing methods has been underlined. The proposed one is an orthogonal
projection-based optimal estimation of the nominal behavior in the
input-output data space $\left( u_{s}(k),y_{s}(k)\right) ,$ and the existing
one is based on an LS-estimation of the output $y_{s}(k).$

In our recent work, a unified framework of control and detection based on
the Bezout identity and the mapping (\ref{eq2-3}) between the input-output $%
\left( u,y\right) $ and the latent variable pair $\left( v,r\right) $ is
established. Inspired by this work, analogous forms of the Bezout identity
and the mapping (\ref{eq2-3}) will be explored in the finite-sample
subspaces and system representations. Their potential application to
data-driven feedback control would be the objective of these research
endeavours.

\end{document}